\documentclass[10pt]{article}
\usepackage{standalone}

\usepackage{datetime}
\usepackage{lmodern}

\usepackage{graphicx}

\newif\ifproofs

\proofstrue

\usepackage[paper=a4paper,bottom=2.5cm,left=2.5cm,right=2.5cm,top=2.5cm]{geometry}
\usepackage[numbers,sort&compress]{natbib}
\usepackage{amsmath,amsfonts,amscd,amssymb,dsfont} 
\usepackage{cleveref}
\usepackage{bm}
\usepackage{textcomp}

\usepackage[shortlabels]{enumitem}

\usepackage{graphics} 
\usepackage{amssymb}  

\usepackage{amsthm}

\usepackage{pgfplots}
\pgfplotsset{compat=1.11}
\usepgfplotslibrary{fillbetween}
\usetikzlibrary{intersections}
\pgfdeclarelayer{bg}
\pgfsetlayers{bg,main}

\newtheorem{theorem}{Theorem}[section]

\newtheorem{proposition}{Proposition}[section]
\newtheorem{lemma}{Lemma}[section]

\newtheorem{assumption}{Assumption}

\theoremstyle{definition}
\newtheorem{definition}{Definition}[section]
\newtheorem*{notation}{Notations}
\newtheorem{remark}{Remark}[section]

\Crefname{corollary}{Corollary}{Corollaries}
\Crefname{equation}{Eq.}{Eqs.}
\Crefname{figure}{Figure}{Figures}
\Crefname{tabular}{Tab.}{Tabs.}
\Crefname{table}{Tab.}{Tabs.}
\Crefname{lemma}{Lemma}{Lemmas}
\Crefname{theorem}{Theorem}{Theorems}
\Crefname{definition}{Definition}{Definitions}
\Crefname{section}{Section}{Sections}
\Crefname{proposition}{Proposition}{Propsitions}
\Crefname{assumption}{Assumption}{Assumptions}
\Crefname{example}{Example}{Examples}

\usepackage[colorinlistoftodos,bordercolor=orange,backgroundcolor=orange!20,linecolor=orange,textsize=scriptsize]{todonotes}
\usepackage{multirow}
\usepackage{pbox}
\usepackage{amsfonts}

\newcommand{\norm}[1]{\left\Vert #1\right\Vert}
\newcommand{\tnorm}[1]{\textstyle\left\Vert #1\right\Vert}

\newcommand{\txt}{\textstyle}

\newcommand{\xx}{\bm{x}}

\newcommand{\yy}{\bm{y}}

\newcommand{\zz}{\bm{z}}
\newcommand{\rr}{\bm{r}}

\newcommand{\hf}{\hat{f}}
\newcommand{\g}{\mathbf{g}}
\newcommand{\h}{\mathbf{h}}

\newcommand{\xag}{X}
\newcommand{\xxag}{\bm{\xag}}
\newcommand{\yag}{Y}
\newcommand{\yyag}{\bm{\yag}}
\newcommand{\zag}{Z}
\newcommand{\zzag}{\bm{\zag}}

\newcommand{\rit}{\mathbb{R}}
\newcommand{\nit}{\mathbb{N}}

\newcommand{\X}{\mathcal{X}}
\newcommand{\FX}{\widetilde{\X}} 
\newcommand{\Sxag}{\overline{\X}} 
\newcommand{\Y}{\mathcal{Y}}

\newcommand{\T}{\mathcal{T}}

\newcommand{\I}{\mathcal{I}}
\newcommand{\G}{\mathcal{G}}
\newcommand{\GA}{\mathcal{G}(A)}

\newcommand{\M}{\mathcal{M}}

\newcommand{\NE}{\mathrm{NE}}
\newcommand{\GNE}{\mathrm{GNE}}

\newcommand{\Gna}{G} 
\newcommand{\GnaA}{G(A)}

\newcommand{\ww}{\bm{w}} 
\newcommand{\wwag}{\bm{W}} 
 
\newcommand{\cc}{\bm{c}} 

\newcommand{\diamX}{M} 

\newcommand{\Bdg}{B_\mathbf{g}}
\newcommand{\Bdf}{B_\mathbf{f}}
\renewcommand{\th}{\theta} 
\renewcommand{\t}{t}
\renewcommand{\i}{i}
\newcommand{\iti}{_{\i,\t}}
\newcommand{\thti}{_{\th,\t}}
\def\ub{\underline{b}}
\def\ob{\overline{b}}

\newcommand{\hxx}{\hat{\xx}}

\newcommand{\hxxag}{\hat{\xxag}}

\newcommand{\hyy}{\hat{\yy}}
\newcommand{\hyyag}{\hat{\yyag}}
\newcommand{\hz}{\hat{\zz}}

\newcommand{\hzz}{\bm{\hat{z}}}

\newcommand{\sxx}{\xx^*}  
\newcommand{\sxxag}{\xxag^*}
\newcommand{\syy}{{\yy^*}} 

\newcommand{\sz}{{\zz^*}}

\newcommand{\eqd}{\triangleq}
\newcommand{\dth}{\,\mathrm{d}\th}

\newcommand{\nn}{\bm{n}}
\newcommand{\snu}{\nu} 
\newcommand{\esnu}{^{\snu}}

\newcommand{\dset}{\delta} 
\newcommand{\mdset}{\overline{\delta}} 

\newcommand{\duti}{d} 
\newcommand{\mduti}{\overline{d}} 
\newcommand{\mmu}{\overline{\mu}} 

\newcommand{\stgccvut}{{\alpha}} 

\newcommand{\disth}{\sigma} 
\newcommand{\cutp}{\upsilon} 

\renewcommand{\ss}{\bm{s}}

\newcommand{\A}{\bm{A}}
\newcommand{\bb}{\bm{b}}
\newcommand{\us}{\underline{s}}
\newcommand{\os}{\overline{s}}
\newcommand{\dimp}{K} 

\newcommand{\mld}{\overline{\lambda}}
\newcommand{\ld}{\lambda}

\newcommand{\rlt}{\text{ri}\,}
\newcommand{\rbd}{\text{rbd}\,}
\newcommand{\spa}{\text{span}\,}
\newcommand{\aff}{\text{aff}\,}
\newcommand{\supk}{^{(k)}}
\newcommand{\bpsi}{\overline{\psi}}

\title{
Nonsmooth Aggregative Games with Coupling  Constraints\\and Infinitely Many Classes of Players 
}

\author{{Paulin Jacquot\thanks{Paulin Jacquot is with  EDF Lab, Inria and  Ecole polytechnique, CNRS. {\tt\small paulin.jacquot@polytechnique.edu} } ~and Cheng Wan \thanks{Cheng Wan is with LMO, Universit\'e Paris-Sud; Inria Paris and RIIS, SUFE. {\tt\small cheng.wan.2005@polytechnique.org } (Corresponding author). ~~~{This work was supported in part by the PGMO foundation.}        }
}
}

\begin{document}

\maketitle
\begin{abstract}
%
After defining a pure-action profile in a nonatomic aggregative game, where players have specific continuous pure-action sets and nonsmooth cost functions, as a square-integrable function, we characterize a Wardrop equilibrium (WE) as a solution to an infinite-dimensional generalized variational inequality. 
We show the existence of variational WE in monotone nonatomic aggregative games with coupling constraints. Its uniqueness is proved for strictly or aggregatively strictly monotone cases. We then show that, for a sequence of finite-player aggregative games with aggregative constraints, if the players' pure-action sets converge to those of a strongly (resp. aggregatively strongly) monotone nonatomic game, and the aggregative constraints in the finite-player games converge to that of the nonatomic game, then a sequence of variational Nash equilibria in these finite-player games converge to the variational WE in pure-action profile (resp. aggregate-action profile). Finally, we show how to construct auxiliary finite-player games for two general classes of nonatomic games.
\end{abstract}


\section{Introduction}
This paper studies firstly the existence and uniqueness of variational Wardrop equilibrium in nonatomic aggregative games with coupling aggregative constraints, where a continuum of players have heterogeneous compact convex pure-action sets and cost functions. It then examines the convergence of a sequence of variational Nash equilibrium in auxiliary finite-player games to the variational Wardrop equilibrium.
\medskip

\noindent \textbf{Background.} 
Aggregative games form a large class of non-cooperative games. In such a game, a player's payoff is determined by her own action and the aggregate of all the players' actions \cite{Corchon1994}. The setting of aggregative games is particularly relevant to the study of nonatomic games \cite{schmeidler1973equilibrium}, games with a continuum of players. There, a player has an interaction with the other players only via an aggregate-level profile of their actions, for example, the distribution of certain actions, while she has no interest or no way to know the behavior of any particular player or the identity of the player making a certain choice.  

Nonatomic games are readily adapted to many situations in industrial engineering or public sectors where a huge number of users, such as traffic commuters and electricity consumers, are involved. 
Those users have no direct interaction with each other except through the aggregate congestion or consumption to which they are contributing simultaneously.
These situations can often be modeled as a congestion game, a special class of aggregative games, both in nonatomic version and finite-player version.
The latter, called atomic congestion game, was formally formulated by Rosenthal in 1973 \cite{rosenthal1973network}, while related research work in transportation and traffic analysis, mostly in the nonatomic version, appeared much earlier \cite{wardrop1952some,Bec56}.
 The theory of congestion games has also found numerous applications in telecommunications \cite{orda1993competitive}, distributed computing \cite{altman2002nash}, energy management \cite{atzeni2013demand}, and so on. 

The concept of equilibrium in nonatomic games is captured by the so called Wardrop equilibrium (WE) \cite{wardrop1952some}. A nonatomic player neglects the impact of her deviation on the aggregate profile of the whole population's actions, in contrast to a finite player. 
For the computation of WE, existing results are limited to particular classes of nonatomic games, such as population games \cite{MaynardSmith82,HofSig98,San11}, where finite types of players are considered, each type sharing the same finite number of pure actions and the same payoff function. 
Convergence of some dynamical systems describing the evolution of pure-action distribution in the population has been established for some particular equilibria in some particular classes such as linear games \cite{TayJon78}, potential games \cite{Bec56,San00} and stable games \cite{smith84stability,HofSan09}. 
Algorithms corresponding to discretized versions of such dynamical systems for the computation of WE have been studied, in particular for congestion games \cite{Friesz1994daytoday,ZhangNa1997}. 
\medskip

\noindent \textbf{Motivation.} 
This paper is mainly motivated by two gaps in the literature on nonatomic games.

Firstly, in engineering applications of nonatomic games such as the management of traffic flow or energy consumption, individual commuters or consumers often have specific choice sets due to individual constraints, and specific payoff functions due to personal preferences. 
Also, unlike for a transportation user who usually choose a single path, an electricity consumer faces to a resource allocation problem where she has to divide the consumption of a certain quantity of energy onto different time periods. Hence, her pure-action set is no longer a finite, discrete set as a commuter but a compact convex set in $\rit^T$ where $T$ is the total number of time periods. 
 Fewer results exist for the computation of WE for the case where players have infinitely many different types (i.e. action sets and payoff functions) or where they have continuous action sets. For example, most of the works in line with Schmeidler  \cite{MASCOLELL1984,rath1992direct,khan1997continuum,Carmona2009large} use fixed-point theorems to prove the existence of WE, though in fairly general settings. Besides, most of the existing work assumes smooth cost functions of players which is somewhat strong.

Secondly, in the above-mentioned applications of aggregative games, coupling  constraints, especially those at aggregative level, commonly exist \cite{gramm2017}. 
 Examples are  capacity constraints of the network or power-grid, and ramping constraints on the variation of total energy consumption between time periods. 
  In this regard, there are even fewer studies in game-theoretical modeling of nonatomic game. However, the presence of coupling constraints adds non trivial difficulties in the analysis of WE and their computation. Indeed, an appropriate definition of equilibrium is already not obvious. 
  An analog to the so-called generalized Nash equilibrium \cite{harker1991gne} for finite-player games does not exist for nonatomic games because a nonatomic player's behavior has no impact on the aggregative profile. 
  Moreover, dynamical systems and algorithms used to compute Wardrop equilibria in population games cannot be straightforwardly extended to this case. 
  Indeed, in these dynamics and algorithms, players adapt their strategies unilaterally in their respective strategy spaces, which can well lead to a new strategy profile violating the coupling constraint.

In view of these two gaps, the main objective of this paper is to provide a model of nonatomic aggregative games with \emph{infinitely many player-specific, compact convex pure-action sets and infinitely many player-specific nonsmooth payoff functions}, 
then introduce a general form of coupling aggregative constraints into these games, choose an appropriate equilibrium notion, study their properties such as existence and uniqueness and, finally, their computation. 
%
\medskip

\noindent \textbf{Main results.}  
%
After defining a pure-action profile in a nonatomic aggregative game where players have specific compact convex pure-action sets lying in $\rit^T$, and specific cost functions, convex  in their own action variable but nonsmooth, 
 \Cref{thm:agg_wardrop} characterizes a Wardrop equilibrium (WE) as a solution to an infinite-dimensional generalized variational inequality (IDGVI). 

\Cref{prop:exist_ve_inf} proves the existence of WE and variational Wardrop equilibrium (VWE), equilibrium notion in the presence of coupling constraints defined by a similar IDGVI, in monotone nonatomic games by showing the existence of solutions to the characteristic IDGVI. 
Then, \Cref{prop:unique_vwe} shows the uniqueness of WE and VWE in case of strictly monotone or aggregatively strictly monotone games. The definition of monotone games is an extension of the stable games \cite{HofSan09}, also called dissipative games \cite{sorin2015finite}, in population games with a finite types of nonatomic players to the case with infinitely many types.

\Cref{thm:converge_with_u} is the main result of this paper. 
It shows that, for a sequence of finite-player aggregative games, if the players' pure-action sets converge to those of a strongly monotone (resp.  aggregatively strongly monotone) nonatomic aggregative game, and if the aggregative constraints in these finite-player games converge to the aggregative constraint in the nonatomic game, then a sequence of so-called variational Nash equilibria (VNE) in these finite-player games converge, in pure-action profile (resp. in aggregate-action profile), to the VWE.
We provide an upper bound on the distance between the VNE and VWE, specified as a function of the parameters of the finite-player and nonatomic games.

This result allows the construction of an auxiliary sequence of finite-player games with finite-dimensional VNE so as to approximate the infinite-dimensional VWE in the special class of strongly or aggregatively strongly monotone nonatomic aggregative games, with or without aggregative constraints. Since there are much more results \cite{facchinei2007finite} on the resolution of finite-dimensional variational inequalities characterizing  VNE, we can therefore obtain an approximation of the VWE with arbitrary precision.  

Finally, we show how to construct an AAS for two general classes of nonatomic games.
\medskip

\noindent \textbf{Related work.} Extensive research has been conducted on Wardrop equilibria in nonatomic congestion games via their formulation with variational inequalities \cite{MarPat2007}, while the similar characterization of Nash equilibria in atomic splittable games, where players have continuous action sets in contrast to unsplittable games where players have finite action sets, has been less studied \cite{haurie1985relationship,orda1993competitive}. 
In addition to their existence and uniqueness, the computational and dynamical aspects of equilibria as solutions to variational inequalities have also been studied  \cite{Smith1984a,ZhuMarcotte1994,ZhangNa1997,CPP2002}. 
However, in most cases, the variational inequalities involved have finite dimensions, as opposed to the case of Wardrop equilibrium in this paper. 
 Marcotte and Zhu \cite{marcotte1997equilibria} consider nonatomic players with continuous types (leading to a characterization of the Wardrop equilibrium as an infinite-dimensional variational inequality) and studied the equilibrium in an aggregative game with nonatomic players differentiated through a linear parameter in their cost function.

Some results have already been given to quantify the relationship	 between Nash and Wardrop equilibria. 
Haurie and Marcotte \cite{haurie1985relationship} show that in a sequence of atomic splittable games where atomic splittable players are replaced by smaller and smaller equal-size players with constant total weight, Nash equilibria converge to the Wardrop equilibrium of a nonatomic game. 
Their proof is based on the convergence of variational inequalities corresponding to the sequence of Nash equilibria, a technique similar to the one used in this paper. Wan \cite{wan2012coalition} generalizes this result to composite games where nonatomic players and atomic splittable players coexist, by allowing the atomic players to replace themselves by players with heterogeneous sizes.

Gentile et al. \cite{gentile2017nash} consider a specific class of finite-player aggregative games with linear coupling constraints. They use the variational inequality formulations for the unique generalized Nash equilibrium and the unique generalized Wardrop-type equilibrium (which consists in letting each finite player act as if she was nonatomic) of the same finite-player game to show that, when the number of players grows, the former can be approximated by the latter. There are several differences between our model and theirs. Firstly, we consider nonatomic games with players of infinitely-many different types instead of finite-player games only. Secondly, we consider variational Nash and Wardrop equilibria instead of generalized equilibria (which does not exist in nonatomic games). In contrast to generalized equilibria, a variational equilibrium is not characterized by a best reply condition for each of the finite or nonatomic players, which makes the study of its properties much more difficult, as shown in \Cref{sec:approx_games}. Thirdly, we allow for nonsmooth cost functions and general form of coupling constraints while they consider differentiable cost functions and linear coupling constraints.
 
Milchtaich \cite{milchtaich2000generic} studies finite and nonatomic crowding games (similar to aggregative games), where players have finitely many pure actions, and shows that, if each player in an $n$-person game is replaced by $m$ identical replicas with constant total weight, pure Nash equilibria generically converge to the unique equilibrium of the limit nonatomic game as $m$ goes to infinity. His proof is not based on a variational inequality formulation.
\medskip

\noindent \textbf{Structure.} The remaining of the paper is organized as follows. \Cref{sec:atomicNonatomicDefs} recalls the definition of finite-player aggregative games with and without aggregative constraints, the notion of equilibrium in these cases and their properties. \Cref{sec:nonatomic} is dedicated to nonatomic aggregative games with and without aggregative constraints. After defining Wardrop equilibrium and variational Wardrop equilibrium, we concentrate on the special class of monotone games and show the existence and uniqueness of equilibria there via generalized infinite dimensional variational inequalities. 
In \Cref{sec:approx_games}, we give the definition of an approximating sequence of finite-player games for a nonatomic aggregative game with or without coupling constraints, and present the main theorem of the paper on the convergence of the sequence of (variational) Nash equilibria of the approximating finite-player games to the (variational) Wardrop equilibrium of the nonatomic game. 
The construction of such a sequence of approximating finite-player games is shown for two important classes of nonatomic games. 
\medskip
 
\noindent \textbf{Notations.} Vectors are denoted by a bold font (e.g. $\xx$) as opposed to scalars (e.g. $x$). 

The transpose of vector $\xx$ is denoted by  
 $\xx^\tau$. 

The closed unit ball in a metric space, centered at $\xx$ and of radius $\eta$, is denoted by $N_\eta(\xx)$. 

For a nonempty convex set $C$ in a Hilbert space $\mathcal{H}$ (over $\rit$), 
\begin{itemize}
\item  $T_C(\xx)=\{\yy\in \mathcal{H}: \yy=0 \text{ or } \exists (\xx_k)_k \text{ in } C \text{ s.t. } \xx_k\not\equiv \xx, \xx_k\rightarrow \xx, \frac{\xx_k - \xx}{\|\xx_k-\xx\|}\rightarrow \frac{\yy}{\|\yy\|}\}$ is the \emph{tangent cone} of $C$ at $\xx\in C$; 
\item $\spa C=\{\sum_{i=1}^k \alpha_i \xx_i : k\in \nit, \alpha_i\in\rit, \xx_i\in C\}$ is the \emph{linear span} of $C$;
\item  $\aff C=\{\sum_{i=1}^k \alpha_i \xx_i : k\in \nit, \alpha_i\in\rit, \sum_i \alpha_i = 1, \xx_i\in C\}$ is the  affine hull of $C$;
\item $\rlt C=\{\xx\in C: \exists \eta>0 \text{ s.t. } N_\eta(\xx) \cap \aff C \subset C\}$ is the \emph{relative interior} of $C$;
\item $\rbd C$ is the \emph{relative boundary} of $C$ in $\mathcal{H}$, i.e. the boundary of $C$ in $\spa C$.
\end{itemize}

The \emph{inner product} of two points $\xx$ and $\yy$ in any Euclidean space $\rit^T$ is denoted by $\langle \xx, \yy \rangle= \sum_{i=1}^T x_i y_i$. The $l^2$-\emph{norm} of $\xx$ is denoted by $\norm{\xx} \eqd \langle \xx, \xx \rangle^{1/2}$. 

We denote by $L^2([0,1], \rit^T)$ the Hilbert space of  measurable functions from $[0,1]$  (equipped with the Lebesgue measure $\mu$) to $\rit^T$ that are square integrable with respect to the Lebesgue measure . The \emph{inner product} of two vector functions $F$ and $G$ is denoted by $\langle F, G \rangle_2 = \int_{0}^1 \langle F(\th), G(\th) \rangle \dth$. The Hilbert space $L^2([0,1], \rit^T)$ is endowed with \emph{$L^2$-norm}: $\|F\|_{2}=\langle F, F \rangle_{2}^{1/2}$.

The distance between a point $\xx$ and a set $A$ is denoted by $d_m(\xx,A)\eqd\inf_{\yy \in A}\norm{\xx-\yy}_m$, where $m$ is omitted or is equal to $2$, depending on whether we consider an Euclidean space or  $L^2([0,1], \rit^T)$. 

Similarly, the \emph{Hausdorff distance} between two sets $A$ and $B$ is denoted by $ d_{H,m} (A, B )$, which is  defined as $\max \{\sup_{\xx\in A}d_m(\xx,B), \sup_{\yy \in B}d_m(\yy,B)\}$. 
Later, we will define new metrics indexed by $\snu$ in Euclidean spaces. The point-set distances and Hausdorff distances are defined similarly and denoted with index $\snu$.


The \emph{subdifferential}, i.e. set of \emph{subgradients} of a convex function $f$ at $\xx\in \rit^T$ in its domain $C$, which is a convex set in $\rit^T$, is denoted by $\partial f(\xx)$. Recall that if vector $\g\in \rit^{T}$ is a subgradient of $f$ at $\xx$, then for all $\zz\in C$, $f(\zz) \geq f(\xx) + \langle \g, \zz-\xx\rangle$.

For a function $(\xx,\xxag)\mapsto f(\xx,\xxag)$ of two explicit variables, convex in $\xx$, we denote by $\partial_1 f(\xx, \xxag)$ the (nonempty) subdifferential of function $f(\cdot,\xxag)$ for any fixed $\xxag$.

\section{Finite-player aggregative games} 
\label{sec:atomicNonatomicDefs}
This section recalls the definition of finite-player aggregative games with and without coupling aggregative constraints, and some notions of equilibrium in these games as well as their characterization by generalized variational inequalities. 


\begin{definition}[Finite-player aggregative game]\label{def:atomicGame} 
A \emph{finite-player aggregative game} is a non-cooperative game specified by:\\
(i) a finite set of players $\I=\{1,\dots ,I\}$,\\
(ii) a set of feasible pure actions  $\X_i \subset \rit^T$  for each player $i$, where $T\in \nit^*$ a constant,  with a typical pure action $\xx_\i=(x\iti)_{t=1}^T\in \X_\i$,\\
(iii) a cost function $\X_i\times \rit^T \rightarrow \rit: f_i(\xx_i, \sum_{j\in I}\xx_j)$ for each player $i$, so that a player's cost is determined by her own action and the aggregate action profile.
\smallskip

Denote  $\FX\eqd \X_1\times\dots \times \X_I$. The pure-action profile of the players $\xx\eqd (\xx_\i)_{\i\in\I}$ induces an aggregate-action profile load attributed to arc $t\in \T$, which is denoted by $\xxag=(\xag_t)_{t=1}^T$ where $\xag_t=\sum_{\i\in\I}x\iti$. 
Denote the set of feasible aggregate-action profiles by $\Sxag\eqd \{ \xxag \in \rit^T  : \exists \xx \in \FX \text{ s.t. }\txt\sum_{\i\in\I} \xx_\i= \xxag  \}$.

Denote the game by $\G=(\I, \FX, (f_i)_{i\in \I})$.
\end{definition}
  The following assumptions and notations are adopted in this paper. 
\begin{assumption}[Convex costs] \label{assp_convex_costs}
For each $\i\in\I$, the function $\hat{f}_i(\xx_i,\xxag_{-i}) \eqd f_i(\xx_i, \xxag)$, where  $\xxag_{-\i}=\sum_{j\in \I, j\neq i}\xx_j$, is continuous in $\xx_\i$ and in $\xxag_{-i}$, and is convex in $\xx_i$ for all $\xx_{-\i}=(\xx_j)_{j\in \I, j\neq i}\in \prod_{j\in \I, j\neq i}\X_j$. 
\end{assumption}
\begin{assumption}[Convex and compact strategy sets] \label{assp_compactness}
For each $\i\in\I$, the set $\X_\i$ is a nonempty, convex and compact subset of $\rit^T$.
\end{assumption}


Recall the definition of Nash Equilibrium in finite-player non-cooperative games. 
\begin{definition}[Nash Equilibrium ($\NE$) \cite{nash1950equilibrium} ] \label{def:NE}
A (pure) \emph{Nash equilibrium} of $\G$ is a profile of pure actions $\hxx \in \FX$ such that $\hf_\i(\hxx_\i, \xxag_{-\i}) \leq \hf_\i(\xx_\i,\xxag_{-\i})$ for all $\xx_i \in \X_i$ and all $\i\in \I$.
\end{definition}

Define a correspondence $H:\FX\rightrightarrows \rit^{IT}$ by 
\begin{equation*}
\forall \xx\in \FX, \; H(\xx) \eqd\{(\g_i)_{i\in \I}\in \rit^{IT}: \g_i \in \partial_1 \hf_\i(\xx_\i, \xxag_{-\i}) , \  \forall i\in \I\} = \prod_{i\in\I} \partial_1 \hf_\i(\xx_\i, \xxag_{-\i}) \ .
\end{equation*}
Since the cost functions are convex in players' own strategies, NE can be characterized as solutions to generalized variational inequalities (GVI) \cite{FangPeterson1982gvi}. 
\begin{proposition}[GVI formulation of NE]\label{prop:GNNash} 
Under \Cref{assp_compactness,assp_convex_costs}, $\hxx\in\FX$ is an $\NE$ of $\G$ if and only if either of the following two equivalent conditions holds:
\begin{subequations} 
\begin{align}
& \forall i\in \I, \exists\, \g_i \in \partial_1 \hf_\i(\hxx_\i, \hxxag_{-\i}) \text{ s.t. } \big\langle \g_\i, \xx_\i- \hxx_\i\big\rangle \geq 0, \; \forall \xx_\i \in\X_\i, \label{cond:ind_opt_N}\\
&\exists\, \g \in H(\hxx) \text{ s.t. } \big\langle \g , \xx - \hxx \big\rangle \geq 0,\; \forall \xx \in\FX\ .\label{cond:ind_opt_Nbis}
\end{align}
 \end{subequations}
\end{proposition}
\begin{proof}
\Cref{cond:ind_opt_N} is a necessary and sufficient condition for $\hxx_\i$ to minimize the convex function $\hf_\i (., \hxxag_{-\i})$ on $ \X_\i$ (\cite[Proposition 27.8]{combettes2011monotone}). The equivalence between \eqref{cond:ind_opt_N} and \eqref{cond:ind_opt_Nbis} is obvious.
\end{proof}
\begin{remark}[Generalized VI and Generalized NE are different things]
The variational inequalities (VI) are of \emph{generalized} type here because the subdifferentials of cost functions are not necessarily singled valued. In the case that cost functions are differentiable with respect to the players' own actions, the GVI are reduced to a VI. Do not confuse with \emph{generalized NE} in \emph{generalized games} (cf. \Cref{def:general-atomicGame} ) which are characterized by \emph{(generalized-)quasi-VI}.
\end{remark}
The existence of an NE is obtained by a classical result in game theory for finite-player continuous game, since the players have convex continuous cost functions and convex compact pure-action sets. No differentiability condition is needed.
\begin{proposition}[Existence of NE, \cite{Debreu1952,Glicksberg1952,Fan1952}]
Under \Cref{assp_compactness,assp_convex_costs},  $\G$ admits an NE.
\end{proposition}
\begin{remark}[NE is a unilateral level stability condition]\label{rm:unilateral_deviation}
The NE condition ensures stability not only in terms of a single player's behavior but also in terms of their collective welfare. Indeed, on the one hand,  condition \eqref{cond:ind_opt_N} is equivalent to $\langle \g_\i, \yy_\i \rangle \geq 0$ for all $\yy_\i\in T_{\X_\i}(\hxx_\i)$ for each $\i$, i.e. a unilaterally feasible deviation of player $i$ increases her cost; on the other hand, condition \eqref{cond:ind_opt_Nbis} is equivalent to $\langle \g, \yy \rangle \geq 0$ for all $\yy \in T_{\FX}(\hxx)$, i.e. a collectively feasible deviation of all the players increases their total costs. The two conditions are equivalent because the players have independent pure-action spaces so that $T_{\FX}(\hxx)=\prod_{\i\in \I}T_{\X_\i}(\hxx_\i)$, i.e. any collectively feasible deviation can be decomposed into unilaterally feasible deviations. This remark is important because it is no longer the case when one introduces a coupling constraint in the game.
\end{remark}

%
%


The coupling aggregative constraint considered in this paper is of the following general form: There is a nonempty, convex and compact subset $A$ of $\rit^T$, whose intersection with $\Sxag$ is not empty, such that the aggregate-action profile $\xxag \in A$. An example is $A=\{\xxag\in \rit^T_+: \underline{M}_t\leq \xag_t\leq \overline{M}_t, \forall t\in \T; a_t\leq \xag_{t+1}-\xag_t \leq b_t, \forall t\in \T\setminus\{T\}\}$.

\begin{definition}[Finite-player aggregative game with aggregative constraints]\label{def:general-atomicGame} 
 Its only difference from the game defined in \Cref{def:atomicGame}  is that, for each player $\i \in \I$, given the profile of pure actions of the others players $\yy_{-i}\in \prod_{j\in \I\setminus\{i\}}\X_j$, her feasible pure-action set becomes $\X_\i(\yy_{-\i})\eqd\{\xx_\i \in \X_i : (\xx_\i, \yy_{-i}) \in \FX(A)\}$,
where $\FX(A)$ is a  subset of $\FX$ defined by
\begin{equation*}
\FX(A)=\left\{\xx \in \FX: \xxag=\txt\sum_{\i\in\I}\xx_i \in A\right\} \ .
\end{equation*}
This game is denoted by $\GA=( \I,  \FX , A, (f_i)_{\i\in\I})$ or simply $\GA$.
\end{definition}

Finite-player non-cooperative games with coupling constraints are called generalized Nash games \cite{harker1991gne}. The extension from games to generalized games is not trivial. 
In the case with no coupling constraint, the pure-action spaces of the players are independent so that any collectively feasible deviation can be decomposed into unilaterally feasible deviation. This property does not always hold with a coupling constraint. 
To see this, we recall the following notion of generalized equilibrium in generalized games.

\begin{definition}[Generalized Nash Equilibrium (GNE), \cite{harker1991gne}]\label{def:gne}
A profile of pure actions $\hxx \in \X$ is a \emph{generalized Nash equilibrium} of $\GA$ if
\begin{equation*}
\hxx_\i \in \X_i(\hxx_{-\i}) \ \text{ and } \
\hf_\i(\hxx_\i,\hxxag_{-\i}) \leq \hf_\i(\xx_\i,\hxxag_{-\i}), \quad \forall \xx_i \in \X_i(\hxx_{-i}), \quad \forall \i\in\I\ . 
\end{equation*}
\end{definition}

 Its characterization by generalized quasi-variational inequalities (GQVI) \cite{chanpang1982gqvip} can be proved as for \Cref{prop:GNNash} .
\begin{proposition}[GQVI formulation of GNE]\label{prop:QVI-GNE} 
Under \Cref{assp_compactness,assp_convex_costs}, $\hxx\in\FX$ is a $\GNE$ of $\GA$ if and only if one of the following two equivalent conditions holds:
\begin{subequations}
\begin{align}
& \forall\i \in \I: \hxx_i \in \FX(\hxx_{-i}) \text{ and } \exists\, \g_i \in \partial_1 \hf_\i(\hxx_\i, \hxxag_{-\i}) \text{ s.t. } \big\langle \g_i, \xx_\i- \hxx_\i\big\rangle   \geq 0,  \; \forall \xx_\i \in\X_i(\hxx_{-i}) \label{cond:ind_opt_N_G}\\
 &\hxx\in \FX(\hxx)\eqd \prod_{\i\in\I}\X_\i(\xx_{-\i})  \text{ and } \exists\, \g \in H(\hxx) \text{ s.t. }   \txt \big\langle \g, \xx- \hxx\big\rangle\geq 0,\; \forall \xx \in \FX(\hxx) \ .\label{cond:ind_opt_Nbis_G} 
 \end{align}
\end{subequations}
\end{proposition}
%
\begin{remark}[Problematics of GNE]\label{rm:collective_deviation}
The notion of GNE can be problematic. Given a pure-action profile $\xx\in \FX(A)$, firstly simultaneous unilateral deviations can lead to a new profile out of $\FX(A)$. Secondly, if $\FX(\xx)\cap \FX(A)$ is a proper set of $\FX(A)$, profiles in $\FX(A)\setminus \FX(\xx)$ are not feasible. 
Indeed, the unilateral stability condition \eqref{cond:ind_opt_N_G} is equivalent to collective stability only among those collective deviations composed by unilaterally feasible deviations, i.e. condition \eqref{cond:ind_opt_Nbis_G}, because $ T_{\X(\hxx)}(\hxx)=\prod_{\i\in \I}T_{\FX_\i(\hxx_{-i})}(\hxx_\i) $. 
Collective deviations towards $\FX(A)\setminus \FX(\xx)$ are not composed by unilaterally feasible deviations while they may effectively decrease the total cost or even each player's cost (cf. \cite{harker1991gne} for an example). 
A GNE can thus lose its stability when such collective deviations are allowed. To answer to this issue, we consider the stronger notion of equilibrium defined below:
\end{remark}

\begin{definition}[Variational Nash Equilibrium (VNE), \cite{harker1991gne}]\label{def:ve-finite}
A solution to the following GVI problem: 
 \begin{equation}\label{cond:ind_opt_ve}
\text{Find } \hxx\in \FX(A) \text{  s.t. } \exists\, \g\in H(\hxx)  \text{ s.t. }   \txt \big\langle \g, \xx- \hxx\big\rangle\geq 0,\; \forall \xx \in \FX(A).
 \end{equation}
 is called a variational Nash equilibrium of  $\GA$. In particular, if  $\Sxag\subset A$,  a VNE is a NE.
\end{definition}

A pure-action profile $\hxx\in \FX(A)$ is a VNE if and only if  a collective deviation to any profile in $\FX(A)$ is not collectively beneficial. Indeed, VNE is also unilaterally stable, because it is a GNE \cite[Theorem 3]{harker1991gne}. 

%
\begin{proposition}[\cite{harker1991gne}]\label{prop:ve_is_gne}
In $\GA$, under \Cref{assp_compactness}, any VNE is a GNE.
\end{proposition}
%
%

\begin{remark}[VNE refines GNE]\label{rm:refine}
VNE can be seen as a refinement of GNE \cite{Kulkarni2012vne}. With a small perturbation of $A$, a GNE which is not a VNE can no longer be an equilibrium.  Harker \cite{harker1991gne} gives an example of a GNE which is not a VNE.

Hence, VNE is adopted in this paper as the equilibrium notion in the presence of aggregative constraints. Moreover, in \Cref{subsec:GWE} it is argued that the notion of generalized equilibrium cannot even be established in nonatomic games.
\end{remark}

\begin{proposition}[Existence of VNE]\label{prop:exist_ve}
  Under \Cref{assp_compactness,assp_convex_costs}, $\GA$  admits a VNE. 
\end{proposition}
\begin{proof}
%
From the convexity of each $\hf_i$, we deduce that $H$ is a nonempty, convex, compact valued, upper hemicontinuous correspondence. Then \cite[Corollary 3.1]{chanpang1982gqvip} shows that the GVI problem \eqref{cond:ind_opt_ve} admits a solution  on the finite dimensional convex compact $\FX(A)$. (In the case that the $\hf_i$ is partially differentiable with respect to $\xx_i$, then the GVI is reduced to a VI, and Lemma 3.1 in \cite{hartman1966} suffices to show the existence of a solution.)
\end{proof}

\section{A Continuum of Players: Nonatomic Framework}\label{sec:nonatomic}
\subsection{Nonatomic aggregative games}
In nonatomic aggregative games considered here, players have compact pure-action sets, and heterogeneous pure-action sets as well as heterogeneous cost function. This model is in line with the  Schmeidler's seminal paper \cite{schmeidler1973equilibrium}, but in contrast to most of the population games studied in game theory \cite{HofSig98,San11} where nonatomic players are grouped into several populations, with players in the same population have the same finite pure-action set  and the same cost function.

\begin{definition}[Nonatomic aggregative game]\label{def:nonatomicGame} 
A \emph{nonatomic aggregative game} $\Gna$ is defined by:\\
 (i) a continuum of players represented by points on the real interval $\Theta=[0,1]$ endowed with Lebesgue measure, \\
 (ii)  a set of feasible pure actions $\X_\th\subset \rit^T$ for each player $\th\in \Theta$,
 with $T\in \nit^*$ a constant, and\\
  (iii) a cost function $\X_\th\times \rit^T \rightarrow \rit: f_\th(\xx_\th, \xxag)$ for each player $\th$,  where $\xxag=(\xag_t)_{t=1}^T$ and $\xag_\t\eqd \int_0^1 \xx\thti \dth$ denotes the aggregate-action profile. 

The set of feasible pure-action profiles is defined by:
\begin{equation*} 
\FX\eqd \left\{ \xx \in L^2([0,1],\rit^{T}) \   : \ \forall\, a.e. \,\th \in \Theta , \xx_\th \in \X_\th \right\}. 
\end{equation*}

Denote the game by $\Gna=(\Theta, \FX, (f_\th)_{\th\in\Theta})$.
\end{definition}
\begin{remark}
The definition of a nonatomic game asks the pure-action profile $\xx$ to be a measurable and integrable function on $\Theta$ instead of simply being a collection of $\xx_\th\in \X_\th$ for $\th\in \Theta$. In other words, a coupling constraint is inherent in the definition of nonatomic games and the notion of WE. This is in contrast to finite-player games.
\end{remark}
The set of feasible aggregate actions is $\Sxag\eqd \{ \xxag \in \rit^T  : \exists\, \xx \in \FX \text{ s.t. } \txt\int_0^1 \xx_\th \dth = \xxag \}$.

Further assumptions are necessary for $\FX$ to be nonempty and for the existence of equilibria to be discussed later.

\begin{assumption}[Nonatomic pure-action sets]\label{ass_X_nonat}
The correspondence $\X: \Theta\rightrightarrows \rit^T, \th \mapsto \X_\th$ has nonempty, convex, compact values and measurable graph $Gr_\X =\{(\th,\xx_\th) \in \rit^{T+1}: \th \in \Theta, \xx_\th \in \X_\th  \}$, i.e. $Gr_\X$ is a Borel subset of $\rit^{T+1}$. Moreover, for all $\th\in \Theta$,  $\X_\th \subset N_{\diamX}(\mathbf{0})$, with $M>0$ a constant.
\end{assumption}
Under \Cref{ass_X_nonat}, a sufficient condition for $\xx$ to be in $L^2([0,1],\rit^{T})$ is that $\xx$ is measurable.  
\begin{notation}
Denote $\M=[0,M+1]^T$. 
\end{notation}

\begin{assumption}[Nonatomic convex cost functions] \label{ass_ut_nonat}  
For all $\th$, $f_\th$ is defined on $(\M')^2$, where $\M'$ is a neighborhood of  $\M$, and is bounded on $\M^2$, and for each  aggregate profile $\yyag\in \M$,\\
(i) function $Gr_\X \rightarrow \rit^T:  (\th,\xx_\th) \mapsto  f_\th(\xx_\th, \yyag)$ is measurable.\\
(ii) for each $\th\in \Theta$, function $\xx_\th \mapsto f_\th(\xx_\th, \yyag)$ is continuous and convex on $\M'$;\\
(iii) There is $\Bdf>0$ such that $\tnorm{\g} \leq \Bdf$ for all subgradients $\g \in \partial _1 f_\th(\xx_\th,\yyag)$ for each $\xx_\th\in \M$, each $\yyag\in \M$, and each $\th\in \Theta$.
\end{assumption}
\begin{remark}
\Cref{ass_ut_nonat}.(iii) implies that $f_\th(\cdot,\cdot)$'s are Lipschitz in the first variable with a uniform Lipschitz constant $\Bdf$ on $\M^2$ for all $\th$. Besides, if $f_\th(\cdot,\yy)$ is differentiable on $\M$, then $\partial_1 f_\th(\xx_\th, \yyag)$ contains one element $\nabla_1 f_\th(\xx_\th,\yyag)$, the gradient of $f_\th(\cdot,\yyag)$ at $\xx_\th$.
%
\end{remark}

Wardrop equilibrium extends the notion of Nash equilibrium in the framework of nonatomic games, where a single player of measure zero has a negligible impact on the others.
\begin{definition}[Wardrop Equilibrium (WE), \cite{wardrop1952some}]
 
A pure-action profile $\sxx\in \FX$ is a \emph{Wardrop equilibrium} of nonatomic game $\Gna$ if
\begin{equation*}
 f_\th (\sxx_\th, \sxxag) \leq f_\th (\xx_\th, \sxxag), \quad\forall \xx_\th \in \X_\th , \ \forall \, a.e.\,  \th \in \Theta\ .
\end{equation*}
\end{definition}
\smallskip

Before characterizing WE by infinite-dimensional GVI, let us introduce some notions and a technical assumption ensuring that the infinite-dimensional GVI is well-defined.

First, define a correspondence $H: L^2([0,1],\M) \rightrightarrows L^2([0,1],\rit^T)$ as follows: 
\begin{equation}\label{eq:sousgra}
H(\xx)\eqd\{\g = (\g_\th)_{\th\in \Theta}\in L^2([0,1],\rit^T): \g_\th \in \partial_1 f_\th(\xx_\th, \txt\int \xx), \forall \,a.e.\,\th\in \Theta\},\quad \forall \xx\in L^2([0,1],\M) .
\end{equation}
In other words, $H(\xx)$ is the collection of measurable (and integrable because of \Cref{ass_ut_nonat}.(iii)) selections of a subgradient for each $\xx_\th$. 
 
Next, define a best-reply correspondence $Br$ from the set of aggregate-action profiles $\Sxag$ to the set of pure-action profiles $\FX$: 
\begin{equation*}
Br(\yyag) \eqd \{ \xx\in \FX : \xx_\th \in \arg\min_{\X_\th} f_\th(\cdot, \yyag), \forall \th\in \Theta\}, \quad \forall \yyag\in \Sxag.
\end{equation*}

 Finally, fix $\yyag\in \Sxag$ and $\xx\in Br(\yyag)$, define a correspondence $\mathcal{D}(\xx,\yyag)$ from $\Theta$ to $\rit^T$ as follows:
 \begin{equation}
 \mathcal{D}(\xx,\yyag) (\th)\eqd \{\g_\th \in \partial_1 f_\th(\xx_\th,\yyag): \langle \g_\th , \yy_\th - \xx_\th \rangle \geq 0, \forall \yy_\th \in \X_\th\},\quad \forall \th\in \Theta.
\end{equation} 
Clearly, this is a nonempty and closed-valued correspondence. 
 
 \begin{assumption}\label{assp:tech}
 For all $\yyag\in \Sxag$ and all $\xx\in Br(\yyag)$, $\mathcal{D}(\xx,\yyag)$ is a measurable correspondence. 
 \end{assumption}

\begin{theorem}[GVI formulation of WE]\label{thm:agg_wardrop}%
Under \Cref{ass_X_nonat,ass_ut_nonat,assp:tech}, $\sxx \in \FX$ is a WE of nonatomic game $\Gna$ if and only if either of the following two equivalent conditions is true:
\begin{subequations}
\begin{align}
\forall \,a.e. \, \th \in \Theta, \; &\exists\, \g_\th\in \partial_1 f_\th(\sxx_\th, \sxxag) \text{ s.t } 
\langle \g_\th, \xx_\th - \sxx_\th \rangle \geq 0, \; \forall \xx_\th \in \X_\th  \ , \label{cond:ind_opt}\\
 &\exists\, \g \in H(\sxx) \text{ s.t } \int_{\Theta} \langle \g_\th, \xx_\th - \sxx_\th \rangle \dth \geq 0, \; \forall \xx\in \FX\ . \label{cond:agg_eq}
 \end{align}
\end{subequations}
%
%
\end{theorem}
We need the following lemma for the proof of \Cref{thm:agg_wardrop}. 
\begin{lemma}\label{lm:bestreply} ~~{ }\\ 
(1) For all $\xx\in L^2([0,1],\M)$, $H(\xx)$ is nonempty. \\
(2) For all $\yyag\in \Sxag$, $Br(\yyag)$ is nonempty.\\
(3) Under \Cref{assp:tech}, for all $\yyag\in \Sxag$ and all $\xx\in Br(\yyag)$, there exists a measurable mapping $\th\mapsto \g_\th$ such that $\g_\th \in \mathcal{D}(\xx,\yyag)(\th)$ for each $\th\in\Theta$.
\end{lemma}
\begin{proof}
%
(1) For each $\th$, the subdifferential $\partial_1 f_\th(\xx_\th, \int\xx)$ is nonempty and compact valued, so that a measurable selection exists according to the compact-valued selection theorem  \cite{aumann1976integration}.

(2) Fix $\yyag\in \Sxag $. A consequence of \Cref{ass_ut_nonat}.(i-ii) is that the function $\Theta \times \M \rightarrow \rit: (\theta, \zz) \mapsto f_\th(\zz,\yyag)$ is a Carath\'eodory function, that is, (i) $f_\cdot(\zz,\yyag)$ is measurable on $\Theta$ for each $\zz\in \M$, and (ii) $f_\th(\cdot,\yyag)$ is continuous on $\M$ for each $\th\in \Theta$.  Thus, according to the measurable maximum theorem \cite[Thm. 18.19]{aliprantis2006infinite} applied to $f_\cdot(\cdot,\yyag)$, there exists a selection $\xx_\th\in \arg\min_{\X_\th} f_\th(\cdot, \yyag)$ such that $\xx$ is a measurable function on $\Theta$.

(3) Because of \Cref{assp:tech}, one can apply the compact-valued selection theorem \cite{aumann1976integration}.
 \end{proof}

\begin{proof}[Proof of \Cref{thm:agg_wardrop}]
Given $\sxxag$, \eqref{cond:ind_opt} is a necessary and sufficient condition for $\sxx_\th$ to minimize the convex function $f_\th (., \sxxag)$ on $ \X_\th$. Condition \eqref{cond:ind_opt} implies condition \eqref{cond:agg_eq} because of \Cref{assp:tech}.

For the converse, suppose that $\sxx\in \FX$ satisfies condition \eqref{cond:agg_eq} but not \eqref{cond:ind_opt}. 
Then there must be a subset $\Theta'$ of $\Theta$ with strictly positive measure such that for each $\th \in \Theta'$, 
$\sxx_\th \notin \Y_\th \eqd \arg\min_{\X_\th}f_\th(\cdot,\sxxag)$.
 In particular, for any $\yy_\th \in \Y_\th$, $\langle \g_\th, \, \yy_\th - \sxx_\th \rangle <f_\th(\yy_\th,\sxxag) - f_\th(\sxx_\th,\sxxag)< 0$. 
 By the same argument as in the proof of \Cref{lm:bestreply},  one can select a $\yy_\th\in \arg\min_{\X_\th}f_\th(\cdot,\sxxag)$ for $\th\in \Theta'$ such that $\Theta' \rightarrow \rit^T: \th \mapsto \yy_\th$ is measurable. 
 By defining $\yy_\th = \sxx_\th$ for $\th\notin\Theta'$, one has $\Theta \rightarrow \rit^T: \th \mapsto \yy_\th$ is measurable and hence belongs to $\FX$. However, $\int_{\Theta} \langle \g_\th, \yy_\th - \sxx_\th \rangle \dth= \int_{\Theta'} \langle \g_\th, \yy_\th - \sxx_\th \rangle \dth <0  $, contradicting \eqref{cond:agg_eq}. 
\end{proof}

\begin{remark}\label{rm:tangent_infinitedim}
Condition \eqref{cond:ind_opt} is equivalent to $\langle \g_\th(\sxx_\th, \sxxag), \yy_\th \rangle \geq 0$ for all $\yy_\th \in T_{\X_\th}(\sxx_\th)$ for each $\th$. The interpretation is the same as for atomic players at NE: no unilateral deviation is profitable. However, since each nonatomic player has measure zero, when considering a deviation in the profile of pure actions, one must let players in a set of strictly positive measure deviate: \eqref{cond:agg_eq} means that the collective deviation of players of any set of strictly positive measure increases their cost. Note that the GVI problem \eqref{cond:agg_eq} has infinite dimensions.  
\end{remark}

The existence of WE is obtained by an equilibrium existence theorem for nonatomic games.
\begin{theorem}[Existence of a WE,  \cite{rath1992direct}]\label{thm:exist_we}
Under \Cref{ass_X_nonat,ass_ut_nonat}.(1), if for all $\th$ and all $\yyag\in \M$, $f_\th(\cdot, \yyag)$ is continuous on $\M$, then the nonatomic aggregative game $\Gna$ admits a WE.
\end{theorem}
\begin{proof} 
The conditions required in Remark 8 in Rath's 1992 paper \cite{rath1992direct} on the existence of WE in aggregate games are satisfied. 
\end{proof}
\begin{remark}
No convexity of $f_\th(\cdot,\yyag)$'s are needed.
\end{remark}

\subsection{Monotone nonatomic aggregative games}\label{subsec:monotone}
For the  uniqueness of WE and the existence of equilibrium notion to be introduced in the next subsection for the case with coupling constraints, let us introduce the following notions of monotone nonatomic games.

\begin{definition}\label{def:mono_non}
The nonatomic aggregative game $\Gna$ is \emph{monotone} if
\begin{equation}\label{cd:mono_non}
\int_{\Theta} \langle \g_\th-\h_\th, \xx_\th - \yy_\th \rangle \dth \geq 0, \quad \forall \xx, \yy \in L^2([0,1],\M) \text{ and } \g\in H(\xx), \h \in H(\yy).
\end{equation}

It is \emph{strictly monotone} if the equality in \eqref{cd:mono_non} holds if and only if $\xx=\yy$ almost everywhere.

It is \emph{aggregatively strictly monotone} if the equality in \eqref{cd:mono_non} holds if and only if $\int \xx=\int \yy$.

It is \emph{strongly monotone} with modulus $\alpha$ if
\begin{equation}\label{cd:strong_mono_non}
\int_{\Theta} \langle  \g_\th-\h_\th, \xx_\th - \yy_\th \rangle \dth \geq \alpha\|\xx-\yy\|^2_2, \; \forall \xx,\yy \in L^2([0,1],\M) \text{ and } \g\in H(\xx),  \h \in H(\yy).
\end{equation}

It is \emph{aggregatively strongly monotone} with modulus $\beta$ if
\begin{equation}\label{cd:strong_agg_mono_non}
\int_{\Theta} \langle  \g_\th-\h_\th, \xx_\th - \yy_\th \rangle \dth \geq \beta\|\txt\int \xx-\txt\int \yy\|^2, \;\forall \xx,\yy \in L^2([0,1],\M) \text{ and } \g\in H(\xx),  \h \in H(\yy).
\end{equation}
\end{definition}
\begin{remark}
\Cref{cd:mono_non} means nothing else but $H$ is a monotone correspondence on $L^2([0,1],\M)$ (cf. \cite[Definition 20.1]{combettes2011monotone} for the definition of monotone correspondence in Hilbert spaces).
%
\end{remark}
\begin{remark}
A recent paper of Hadikhanloo \cite{Hadikhanloo2017} generalizes the notion of stable games in population games \cite{HofSan09} to monotone games in anonymous games, an extension of population games with players having heterogeneous compact pure-action sets but the same payoff function. 
 He defines the notion of monotonicity directly on the distribution of pure-actions among the players instead of pure-action profile as we do. The two approaches are compatible.
\end{remark}

Examples of aggregative games are given by cost functions of the form:
\begin{equation}\label{eq:common_form_cost}
f_\th(\xx_\th,\xxag)= \langle \xx_\th , \cc(\xxag) \rangle - u_\th(\xx_\th) \ .
\end{equation}
 Here $\cc(\xxag)$ specifies the per-unit cost (or negative of per-unit utility) of each of the $T$ ``public products'', which is a function of  the aggregative contribution $\xxag$ to each of the ``public products''. Player $\th$'s cost (resp. negative of utility) associated to these products is scaled by her own contribution $\xx_\th$.  The function $u_\th(\xx_\th)$ measures the private utility of player $\th$ (resp. negative of private cost) for the contribution $\xx_\th$. 

For instance, in a public goods game, $-c_t (\xag_t)$ is the common per-unit payoff for using public good $t$, determined by the total contribution $\xag_t$, while $-u_\th(\xx_\th)$ is player $\th$'s private cost of supplying $\xx_\th$ to the public goods;  in a Cournot competition, $-c_t (\xag_t)$ is the common market price for product $t$, determined by its total supply $\xag_t$, while $-u_\th(\xx_\th)$ is player $\th$'s private cost of producing $x_{\thti}$ unit of product $t$ for each product $t$; in a congestion game, $c_t(\xag_t)$ is the common per-unit cost for using arc $t$ in a network, determined by the aggregate load $\xag_t$ on arc $t$, while $u_\th(\xx_\th)$ is player $\th$'s private utility of her routing or energy consuming choice $\xx_\th$.

\begin{proposition}\label{lm:monotonemap}
Under \Cref{ass_X_nonat,ass_ut_nonat,assp:tech}, in a nonatomic aggregative game $\Gna$ with cost functions of form \eqref{eq:common_form_cost}, assume that $\cc$ is monotone on $\M$ and, for each $\th$, $u_\th$ is a concave function on $\M$. Then:\\
(1)  $\Gna$  is a monotone game.\\
(2) If $u_\th$ is strictly concave on $\M$ for all $\th \in \Theta$, then $\Gna$ is a strictly monotone game.\\
(3) If $\cc$ is strictly monotone on $\M$, then $\Gna$ is an aggregatively strictly monotone game.\\
(4) If $u_\th$ is strongly concave on $\M$ with modulus $\alpha_\th$ for each $\th \in \Theta$ and $\inf_{\th\in\Theta}\alpha_\th= \alpha>0$, then $\Gna$ is a strongly monotone game with modulus $\alpha$.\\
(5) If  $\cc$ is strongly monotone on $\M$ with $\beta$, then $\Gna$ is an aggregatively strongly monotone game with modulus $\beta$.
\end{proposition}
\begin{proof}
(1) Let $\xx,\yy\in \FX$ and $\xxag=\int \xx$, $\yyag=\int \yy$. For each $\th$, $\partial_1 f_\th(\xx_\th, \yyag)=\{\cc(\yyag)+\g:\g\in \partial (-u_\th)(\xx_\th) \} $.  
Then, given $\xx,\yy\in \FX$, with $\g_\th \in  \partial (-u_\th)(\xx_\th)$ and $\h_\th \in  \partial (-u_\th)(\yy_\th)$ for each $\th$, one has $\langle  \g_\th(\xx_\th)-\h_\th (\yy_\th), \xx_\th - \yy_\th \rangle\geq 0$ because $u_\th$ is concave so that $\partial (-u_\th)$ is a monotone correspondence on $\X_\th$.

Then 
$\int_0^1 \langle (\cc(\xxag)+\g_\th(\xx_\th))-(\cc(\yyag)-\h_\th(\yy_\th)), \xx_\th-\yy_\th\rangle\dth= \langle \cc(\xxag)-\cc(\yyag), \xxag-\yyag\rangle + \int_0^1\langle  \g_\th(\xx_\th)-\h_\th (\yy_\th), \xx_\th - \yy_\th \rangle \dth\geq 0$ because $\cc$ is monotone. Hence $\Gna$ is a monotone game. 

The proof for (2)-(5) is omitted.
\end{proof}

In particular, if $\cc(\xxag)=(c_t(\xag_t))_{t\in T}$, then $\cc$ is monotone if $c_t$'s are all non-decreasing, and $\cc$ is strongly monotone if $c_t$'s are all strictly increasing.

\subsection{Nonatomic aggregative games with aggregate constraints}\label{subsec:GWE}
Let us consider the aggregative constraint in nonatomic aggregative game $\Gna$: $\xxag\in A$,  where $A$ is a convex compact subset of $\rit^T$ such that $A\cap \Sxag \neq \emptyset$. Let $\FX(A)$ be a subset of $\FX$ defined by $\FX(A)\eqd \{\xx \in \FX: \xxag=\int \xx \in A\}$. Denote the nonatomic game with aggregative constraint $\GnaA$.

In contrast to games with finitely many players, a generalized equilibrium in the style of \Cref{def:gne} is not well-defined in a nonatomic game. Indeed, since the impact of a nonatomic player's choice on the aggregative profile is negligible, the feasible pure-action set of a nonatomic player $\th$ facing the choices of the others $\xx_{-\th}$ in a game with coupling constraint is not a well-established notion: either $\int \xx_{-\th} \in A$ then $\X_{\th}=\X$, or $\int \xx_{-\th} \notin A$ then $\X_{\th}=\emptyset$. 
Departing from a pure-action profile in $\FX(A)$, simultaneous unilateral deviations by the players can lead to any profile in $\FX$. If only profiles in $\X(A)$ are allowed to be attained, then one lands on a notion similar to VNE.
 Indeed, the most natural notion of equilibrium with the presence of aggregative constraint is the notion of variational Wardrop equilibrium, where feasible deviations are defined on a collective basis.

\begin{definition}[Variational Wardrop Equilibrium (VWE)]\label{def:ve-infinite}
A solution to the following infinite dimensional GVI problem: 
 \begin{equation}\label{cond:ind_opt_ve_inf}
\text{Find } \sxx\in \FX(A) \text{ and } \g\in H(\sxx)\;\text{  s.t.  } \int_{\Theta} \langle \g_\th, \xx_\th - \sxx_\th \rangle \dth \geq 0,\quad \forall \xx \in \FX(A),
 \end{equation}
 is called a \emph{variational Wardrop equilibrium} of  $\GnaA$, where the correspondence $H$ is as defined by \Cref{eq:sousgra}.
\end{definition}

\begin{remark}[Justification for VWE]
From a game theoretical point of view, the notion of VWE can be problematic as well. Each nonatomic player can deviate unilaterally without having any impact on the aggregate profile. The unilateral stability as for NE and VNE (cf. \Cref{rm:unilateral_deviation,rm:refine}) is lost. However, our main theorem, \Cref{thm:converge_with_u}, shows that VWE can be seen as the limit of a sequence of VNE which are unilaterally stable.
Finally, in the literation of congestion games, the equilibrium notion characterized by VI of form \eqref{cond:ind_opt_ve_inf} but in finite dimension and with smooth cost functions has long been studied. For example, see \cite{larsson1999side,marcotte2004capacitated,correa2004capacitated,zhongal2011} and references therein.
\end{remark}

The following facts are needed for later use. Under \Cref{ass_X_nonat}:\begin{itemize}
\item  $\FX$ is a nonempty, convex, closed and bounded subset of  $L^2([0,1], \rit^T)$; 
\item $\FX(A)$ is a nonempty, convex and closed subset of $\FX$; \item $\Sxag$ and $A\cap \Sxag$ are nonempty, convex and compact subsets of $\rit^T$.
 \end{itemize} We omit the proof and only point out that $\FX$ and $\Sxag$ are nonempty because of \Cref{ass_X_nonat} and the measurable selection theorem of Aumann \cite{aumann1969measure}, while aggregate-action set $\Sxag$ is compact by \cite[Theorem 4]{aumann1965integral}.

 \Cref{prop:exist_ve_inf} shows the existence of VWE via the VI approach. Compared with \Cref{thm:exist_we}, much stronger conditions are required on cost functions. 

\begin{assumption}[Continuity of cost function in aggregate action]\label{assp_lip}
For each $\th\in \Theta$ and $\xx_\th\in\M$, $f_\th(\xx_\th,\cdot)$ is continuous on $\M$.
\end{assumption}

\begin{theorem}[Existence of VWE]\label{prop:exist_ve_inf}
Under \Cref{ass_X_nonat,ass_ut_nonat,assp_lip}, if nonatomic game with coupling constraint $\Gna(A)$ is monotone on $\FX(A)$, then a VWE  exists. 
\end{theorem}
\begin{proof}
%

Let us apply \cite[Corollary 2.1]{ding1996monotoneGVI} to show that \Cref{cond:ind_opt_ve_inf} has a solution. This theorem states that if $\FX(A)$ is bounded, closed and convex in $L^2([0,1],\rit^T)$, and if $H:L^2([0,1],\M) \rightrightarrows L^2([0,1],\rit^T)$ is a monotone correspondence which is upper hemicontinuous from the line segments in $\FX(A)$ to the weak* topology of $L^2([0,1],\rit^T)$, then \eqref{cond:ind_opt_ve_inf} admits a solution. 
We only need to show the upper hemicontinuity property. 
First notice that $H(\cdot)$ has closed values.
 Take $\xx$ and $\yy$ in $\FX(A)$, consider sequence $(\xx\supk)_k$ with $\xx\supk=\xx+\frac{1}{k}(\yy-\xx)$, and sequence $(\g\supk)_k$ such that $\g\supk\in H(\xx\supk)$ and $\g\supk \stackrel{\ast}{\rightharpoonup} \g$ with $\g\in L^2([0,1],\rit^T)$. Let us show that $\g\in H(\xx)$. 

Denote $\xxag=\int \xx$ and $\xxag\supk=\int\xx\supk$. Then $\xxag\supk$ converges to $\xxag$ in $l^2$-norm.

By definition of $H$, for each $\zz\in \M$, for each $\th$, $f_\th (\zz_\th, \xxag\supk_\th) \geq f_\th (\xx\supk_\th, \xxag\supk) + \langle \g\supk_\th, \zz_\th-\xx\supk_\th \rangle$. Since $f_\th$ is continuous in both variables, $f_\th (\zz_\th, \xxag\supk)\rightarrow f_\th (\zz_\th, \xxag)$ and $f_\th (\xx_\th\supk, \xxag\supk)\rightarrow f_\th (\xx_\th, \xxag)$. Besides, $\langle \g\supk_\th, \zz_\th-\xx\supk_\th \rangle = \langle \g\supk_\th, \zz_\th-\xx_\th \rangle + \langle \g\supk_\th, \xx_\th-\xx\supk_\th \rangle$, and $\langle \g\supk_\th, \zz_\th-\xx_\th \rangle \rightarrow \langle \g_\th, \zz_\th-\xx_\th \rangle$ because  $\g\stackrel{\ast}{\rightharpoonup} \g$, while $ \langle \g\supk_\th, \xx_\th-\xx\supk_\th \rangle\rightarrow 0$ because $\g\supk_\th$'s are uniformly bounded by $\Bdf$. Therefore, $f_\th (\zz_\th, \xxag) \geq f_\th (\xx_\th\supk, \xxag) + \langle \g_\th, \zz_\th-\xx_\th \rangle$ so that $\g_\th\in \partial_1 f_\th(\xx_\th, \xxag)$. Since the limit of measurable functions is measurable,  $\g$ is measurable. Hence $\g\in H(\xx)$, which concludes the proof.
\end{proof}

\begin{theorem}[Uniqueness of VWE]\label{prop:unique_vwe}
Under \Cref{ass_X_nonat,ass_ut_nonat}: \\
(1) if $\Gna(A)$ is strictly monotone on $\FX(A)$, then it has at most one VWE; \\
(2) if $\Gna(A)$ is aggregatively strictly monotone on $\FX(A)$, then all VWE of $\Gna(A)$ have the same aggregative profile;\\
(3) if $\Gna$ (without aggregative constraint) is only aggregatively strictly monotone but, for each $\th\in \Theta$ and all $\yyag\in \M$, $f_\th(\xx,\yyag)$ is strictly convex in $\xx$, then there is at most one WE.
\end{theorem}
\begin{proof}
Suppose that $\xx,\yy\in \FX(A)$ are both VWE. Let $\xxag=\int\xx$ and $\yyag=\int \yy$. According to Theorem \ref{thm:agg_wardrop}, there exist $\g\in H(\xx)$ an $\h\in H(\yy)$ such that $ \int_{\Theta} \langle \g_\th, \yy_\th - \xx_\th \rangle \dth \geq 0$ and $\int_{\Theta} \langle \h_\th, \xx_\th - \yy_\th \rangle \dth \geq 0$. Adding up these two inequalities yields $\int_{\Theta}\langle \g_\th - \h_\th, \yy_\th - \xx_\th \rangle \dth \geq 0$. 

(1) If $\Gna(A)$ is a strictly monotone game, then $\int_{\Theta}\langle  \g_\th - \h_\th, \xx_\th - \yy_\th \rangle \dth = 0$ and thus $\xx=\yy$ almost everywhere.  

(2-3) If $\Gna(A)$ is an aggregatively strictly monotone game, then $\int_{\Theta}\langle  \g_\th - \h_\th, \xx_\th - \yy_\th \rangle \dth = 0$ and thus $\xxag=\yyag$.  

If there is no aggregative constraint and $f_\th(\cdot,\zzag)$ is strictly convex for all $\zzag\in \M$, then for all $\th$, $\xx_\th$ (resp. $\yy_\th$) is the unique minimizer of $f_\th(\cdot, \xxag)$ (resp. $f_\th(\cdot, \yyag)$). Since $\xxag=\yyag$, one has $\xx_\th=\yy_\th$.
\end{proof}

\section{Approximating a Nonatomic Aggregate Game}
\label{sec:approx_games}
\subsection{Atomic Approximating Sequence}

After discussing the existence and uniqueness properties of NE/VNE and WE/VWE  in finite-player and nonatomic aggregative games, we shall study the relationship between these notions.  First note that NE (resp. WE) is a particular case of VNE (resp. VWE) when the  aggregate constraint set  is  any subset of $\rit^T$ containing $\Sxag$.   On the one hand, one may naturally expect that, when the number of players grows very large in a finite-player aggregative game, the game gets ``close'' to a nonatomic aggregative game.    On the other hand, since there is a larger literature on algorithms for finite dimensional VI than for infinite dimensional ones, it can be helpful for the computation of WE/VWE to find a NE/VNE that approximates the former with arbitrary precision.%

This section shows the following result: Considering a sequence of equilibria of ``approximating'' finite-player aggregative games $(\G\esnu(A\esnu))_\snu$ of a nonatomic game $\Gna(A)$,  where each player in $ \G\esnu(A\esnu) $ represents a collection of nonatomic players who are similar in their action sets and cost functions,
 a sequence of VNE in $(\G\esnu(A\esnu))_\snu$ 
 converges to the VWE of $\GnaA$ when this one is (aggregatively) strongly monotone.

In this section, we always consider that  \Cref{assp_convex_costs,assp_compactness,ass_X_nonat,ass_ut_nonat,assp:tech,assp_lip} hold. 
\medskip

Let us consider the following definition of an approximating sequence:

\begin{definition} \label{def:approx_seq} \textbf{Atomic Approximating Sequence (AAS)} \\
A sequence of finite-player aggregative games $\{\G\esnu(A\esnu)=\big(\I\esnu,\FX\esnu,(f_i\esnu)_i,A \esnu\big): \snu\in \nit^*\}$ with aggregative constraints is an \emph{atomic approximating sequence} (AAS) for the nonatomic aggregative game $\Gna(A)=\big(\Theta,\FX,(f_\th)_\th, A\big)$ with an aggregative constraint if, for each $\snu \in \nit^*$, there exists a partition $(\Theta_0\esnu,\Theta_1\esnu, \dots, \Theta_{I\esnu}\esnu)$ of the set $\Theta$, where $I\esnu\eqd|\I\esnu|$, such that the Lebesgue measure of $\Theta_0\esnu$ is $ \mu_0\esnu= 0$, and if, for each $\i\in \I\esnu$, the Lebesgue measure of $\Theta_\i\esnu$ is $\mu\esnu_i>0$ while the collection of nonatomic players in $\Theta_\i\esnu$ corresponding to finite player $i\in \I\esnu$ satisfies that, as $\nu \rightarrow +\infty$:
\begin{enumerate}[i),leftmargin=*,wide,labelindent=5pt]
\item  \label{def:approx_seq:subgradientAtomicNonatomic} $\mdset\esnu\eqd\max_{\i \in\I\esnu} \dset_\i\esnu \longrightarrow 0$, where $\dset_\i$ is the Hausdorff distance  between the feasible pure-action sets of nonatomic players in $\Theta\esnu_\i$ and the scaled feasible pure-action set of finite player $i\in \I\esnu$:
\begin{equation} \label{eq:def_dset}
\dset_\i\esnu \eqd \sup_{\th \in \Theta\esnu_\i} d_{H}\left( \X_\th, \txt\frac{1}{\mu_i\esnu}\X_{\i}\esnu \right)\ ,
\end{equation}
and $\spa \X\esnu_i = \spa \X_\th$ for all $\th\in \Theta\esnu_i$.
\item $\mld\esnu\eqd \max_{\i\in\I\esnu}\ld_\i\esnu\longrightarrow 0$, where $\ld\esnu_\i$ measures the difference between the subdifferential of a finite player's cost function when she takes or not into account the impact of her own action on the aggregate profile:
\begin{equation}\label{eq:def_ld}
\ld\esnu_\i \eqd \sup_{(\xx,\yyag)\in \M^2} \sup_{\g\in \partial_1 \hf\esnu_\i(\mu\esnu_\i\xx, \yyag-\mu\esnu_\i\xx)}d \left(\g, \partial_1 f\esnu_\i(\mu\esnu_\i\xx, \yyag)\right)
\end{equation}
(Recall that $\hf\esnu_i(\xx_i,\xxag_{-i})\eqd f\esnu_i(\xx_i,\xxag)$. The subgradients of $\hf\esnu_\i$ with respect to the first variable count the impact of $\xx_\i$ on $\xxag$ while the subgradients of $f\esnu_\i$ with respect to the first variable do not.) \label{def:approx_seq:subgradientWithWithoutI}
\item \label{def:approx_seq:distSets}$\mduti\esnu\eqd\max_{\i\in\I\esnu}\duti_\i\esnu \longrightarrow 0  $, where $\duti\esnu_\i$ measures the Hausdorff distance between the subdifferential of nonatomic players' cost functions and that of the finite players' cost functions:
\begin{equation}\label{eq:def_dut}
\duti\esnu_i \eqd  \sup_{\th \in \Theta_\i}\sup_{(\xx,\yyag)\in \M^2}   d_H\left(  \partial_1 f\esnu_\i(\txt {\mu_i\esnu } \xx,\yyag) , \partial_1 f_\th(\xx, \yyag) \right).
\end{equation}
\item \label{def:approx_seq:distAggSets} $D\esnu \longrightarrow 0$, where $D\esnu\eqd d_H\left(A\esnu,A \right)$  is the Hausdorff distance between the aggregative constraint set $A\esnu\subset \rit^T$ and the aggregative constraint set $A\subset \rit^T$. Besides, $\spa A=\spa A\esnu$ for all $\esnu\in \nit^*$.
\end{enumerate} 
\end{definition}
\begin{remark}
Roughly speaking, along an AAS, 
the impact of a finite player's action on the aggregate profile gradually disappears. Besides, the pure-action set of a finite player converges to the pure-action set of a nonatomic player in the subset of $\Theta$ that the finite player represents, whereas the subgradients of her cost with respect to her own action also tend to those of a nonatomic player that she represents.  Finally, the aggregate-profile constraint sets in the AAS converge to the one in the nonatomic game.

Note that, except the last condition on $D\esnu$, the other conditions are independent of the constraint sets $(A\esnu)_\snu$ and $A$.
\end{remark}

\begin{remark}
Without loss of generality, we assume $\rlt (A\cap \Sxag)\neq \emptyset$ in this section. Indeed, if the nonempty convex compact set $A\cap \Sxag$ has an empty relative interior, then it is reduced to a point hence the problem becomes trivial. 
\end{remark}


In  \Cref{subsec:construction}, we will construct an AAS for two fairly general cases of nonatomic games.

%


\medskip

In order to compare a pure-action profile in a finite-player game and one in a nonatomic game, we introduce the following linear mappings which define an equivalent nonatomic action profile for a finite-player action profile and vice versa. 

First, define  $\psi\esnu: \M^{ I\esnu} \rightarrow L^2([0,1],\M)$ for each $\snu\in \nit^*$ by 
\begin{equation}\label{eq:psi}
\forall \xx\esnu\in \M^{I\esnu},\; \psi\esnu(\xx\esnu)=(\psi\esnu_\th(\xx\esnu))_{\th\in \Theta}\, , \text{ where } \,\psi\esnu_\th(\xx\esnu)\eqd\frac{\xx\esnu_i}{\mu\esnu_i}\,, \, \forall i\in\I\esnu ,\  \forall \th\in \Theta\esnu_i\ .
\end{equation}
The interpretation of $\psi\esnu$ is the following. If one considers finite player $i\in \I\esnu$ as a coordinator of the coalition formed by a group of nonatomic players represented by set $\Theta_i\esnu$ equipped with Lebesgue measure, then $\psi_\th\esnu(\xx\esnu)$ dictates the behavior of each of them.  


Then, define mapping $\bpsi\esnu: L^2([0,1],\M) \rightarrow   \M^{\I\esnu} $  for each $\snu\in \nit^*$ by 
 \begin{equation}\label{eq:psibar}
 \forall \xx\in L^2([0,1],\M^T), \;
 \bpsi\esnu(\xx)= \big(\,\bpsi\esnu_i(\xx) \big)_{\i\in\I\esnu}, \, \text{ where }\bpsi\esnu_i(\xx)=\txt\int_{\Theta\esnu_\i} \xx_\th \dth \ .
 \end{equation}
 The interpretation of $\bpsi\esnu$ is that the nonatomic players in $\Theta_i\esnu$ form a coalition which behaves like a finite player, so that $\bpsi\esnu_i(\xx)$ is just her pure-action.

Finally, let us make the following assumption for this section.
\begin{assumption}\label{asp:intpt}
There is a strictly positive constant $\eta$ and a pure-action profile $\bar{\xx}\in \FX$ such that, for almost all $\th\in \Theta$, $d(\bar{\xx}_\th, \rbd \X_\th)>\eta$.
\end{assumption}
It means that the pure-action space of each player has an (aggregatively) nonempty relative interior and that the relative interior is not vanishing along any sequence of players.

\subsection{Convergence of Equilibrium Profiles and Aggregate Equilibrium Profiles}\label{sec:atomic_with_u}

The following \Cref{thm:converge_with_u} gives the main result of this paper. It shows that a VWE in a strongly monotone nonatomic aggregative game can be approximated by VNE of an AAS, both in the case with and without aggregative constraints. 

Recall that, according to \Cref{prop:unique_vwe}, a strongly monotone game is strictly monotone, hence the VWE is unique, while an aggregatively strongly monotone game is aggregatively strictly monotone, hence the aggregate-action profile at VWE is unique. 
\begin{theorem}[Convergence of VNE to VWE] \label{thm:converge_with_u}
Under \Cref{assp_convex_costs,assp_compactness,ass_X_nonat,ass_ut_nonat,assp:tech,assp_lip,asp:intpt}, let $(\G\esnu(A\esnu))_\snu$ be an AAS of nonatomic aggregative game $\Gna(A)$ with an aggregative constraint. 
 Let  $\sxx$ be the VWE of $\GnaA$, $\hxx\esnu \in \FX\esnu(A\esnu)$ a VNE of $\G\esnu(A\esnu)$  for each $\snu\in\nit^*$, and $\sxxag$, $\hxxag\esnu$ their respective aggregate-action profiles. 
 Then there exists constants $\rho>0$ and $\bar{\rho}>0$ such that the following results hold with $K_A\eqd\tfrac{M+1}{\min\{\rho, \bar{\rho}\}}$:

(1) If $\Gna$ is aggregatively strongly monotone with modulus $\beta$, $(\hxxag\esnu)_\snu$ converges to $\sxxag$: for all $\snu\in \nit^*$ such that $\max(\mdset\esnu, D\esnu) < \min\{\bar{\rho}, \rho\}$,
\begin{equation} \label{eq:cvg_agg_nou}
\| \hxxag\esnu- \sxxag \|^2 \leq \frac{1}{\beta}  \Big(  (3\Bdf + 1) K_A\max(D\esnu,\mdset\esnu) +(2M+1)(\mduti\esnu + \mld\esnu) \Big)\ .
\end{equation}

(2) If $\Gna$ is strongly monotone with modulus $\alpha$, then 
 $(\psi\esnu(\hxx\esnu))_\snu$ (cf. \Cref{eq:psi}),  converges to $\sxx$ in $L^2$-norm: for all $\snu\in \nit^*$ such that $\max(\mdset\esnu, D\esnu) < \min\{\bar{\rho}, \rho\}$,
 \begin{equation} \label{eq:cvg_indiv}
\|\psi\esnu(\hxx\esnu)-\sxx\|^2_2\leq \frac{1}{ \stgccvut} \Big( (3\Bdf +1) K_A\max(D\esnu,\mdset\esnu)  +(2M+1)(\mduti\esnu + \mld\esnu)\Big)\ .
\end{equation}

If there are no aggregate constraints, replace $K_A$ and $D\esnu$ all by 0 in \eqref{eq:cvg_agg_nou} and \eqref{eq:cvg_indiv}.
\end{theorem}

Some notions and a series of lemmas are needed for the proof of \Cref{thm:converge_with_u}.
\medskip

Firstly, for $\snu\in \nit^*$, define a new norm $\|\cdot \|_\snu$ on $\rit^{T\I\esnu}$ by $\|\xx\esnu\|_\snu \eqd \left(\sum_{\i\in \I\esnu} \frac{\|\xx\esnu_\i\|^2}{\mu\esnu_i}\right)^{\frac{1}{2}}= \langle \xx\esnu,D_{\mu\esnu} \xx\esnu\rangle^{\frac{1}{2}}$, where  $D_{\mu\esnu}  $ is the symmetric positive definite matrix diag$(\frac{1}{\mu_1\esnu}I_{T}, \dots , \frac{1}{\mu_{N\esnu}\esnu}I_{T})$, with $I_{T}$ being the $T$-dimensional identity matrix. Denote by $d_\snu$ the associated metric, and by $\text{diam}_\snu$ the diameter of a set with norm $\|\cdot \|_\snu$. 

It is easy to see that $\|\xx\esnu\|\leq \sqrt{\mmu\esnu} \|\xx\esnu\|_\snu$, where $\mmu\esnu \eqd\max_{\i \in\I\esnu}\mu_i\esnu$.

The advantage of using $\snu$-norm for game $\G\esnu$ is that it keeps the norm of pure-action profiles at the same order of those in the nonatomic game, instead of $ \sqrt{\mmu\esnu}$ times smaller. Indeed, $\psi\esnu$ is an isomorphism from $\FX\esnu$ (with $\snu$-norm) to $\psi\esnu(\FX\esnu)$ (with $L^2$-norm):
\begin{equation*}
\|\psi\esnu(\xx\esnu)\|^2_2=\int_{\Theta}\|\psi\esnu_\th(\xx\esnu)\|^2\dth=
\sum_{\i\in \I\esnu}\int_{\Theta\esnu_\i}\|\frac{\xx\esnu_i}{\mu\esnu_i}\|^2\dth=\sum_{\i\in \I\esnu}\frac{\|\xx\esnu_i\|^2}{\mu\esnu_i}=\|\xx\esnu\|^2_\snu .
\end{equation*} 

Besides, since $\X_\i\esnu$'s, $\X_\th$'s, $\FX\esnu(A\esnu)$'s and $\FX(A)$ are all convex and closed in their respective Hilbert spaces, the projection functions onto these sets are well defined.
\begin{notation} Let $\Pi\esnu_i(\cdot)$ denote the  projection function onto $\frac{1}{\mu\esnu_i}\X_\i\esnu$ for $\i\in\I\esnu$ and $\Pi_\th(\cdot)$ the projection function onto $\X_\th$ for $\th \in \Theta$. 

Let $\Pi\esnu$ denote the  projection function onto $\FX\esnu(A\esnu)\subset \rit^{I\esnu T}$, and $\Pi$ the projection function onto $\FX(A) \subset L^2([0,1], \rit^T; \mu)$.
\end{notation}

The following \Cref{lem:norm_xat_approx} shows that the players become infinitesimal along an ASS. When scaled by $\frac{1}{\mmu\esnu}$, they stay at the same order as a nonatomic player.
\begin{lemma}\label{lem:norm_xat_approx}
Under \Cref{ass_X_nonat}, for all $\snu \in \nit^*$,  $\|\xx\esnu\|_\snu\leq \mdset\esnu+M$ and $\|\xx\esnu\|\leq \sqrt{\mmu\esnu}(\mdset\esnu  + M)$ for all $\xx\esnu\in \FX\esnu$.
\end{lemma}
\begin{proof}
\ifproofs
Let $\xx\esnu_\i \in \X_\i\esnu$ and $\th \in \Theta_\i\esnu$.  By definition of $\dset_\i\esnu$,  $ \big\|\frac{\xx\esnu_i}{\mu_\i\esnu} -\Pi_\th\big(\frac{\xx\esnu_i}{\mu_\i\esnu}\big)\big\| \leq  \dset_\i\esnu$ so that $\norm{\xx\esnu_i} \leq \mu_i\esnu \big( \dset_\i\esnu  + \big\|\Pi_\th(\frac{\xx\esnu_i}{\mu_\i\esnu})\big\| \big) \leq \mu_i\esnu (\dset_\i\esnu  + M) $. Then, $\|\xx\esnu\|^2_\snu=\sum_{i=1}^{I\esnu}\frac{\|\xx\esnu_i\|^2}{\mu\esnu_i}\leq\sum_{i=1}^{I\esnu} \mu_i\esnu (\dset_\i\esnu  + M)^2\leq  (\mdset\esnu  + M)^2$, and hence $\|\xx\esnu\|^2\leq \mmu\esnu(\dset\esnu  + M)^2$.
\else
Apply the triangle inequality to $\norm{\frac{\xx_i}{\mu_\i\esnu} -P_{\X_\th}\Big(\frac{\xx_i}{\mu_\i\esnu}\Big)}$ where $P_{\X_\th}$ is the projection on $\X_\th$ for a $ \th \in \Theta_i$. 
\fi
\end{proof}

The following lemma shows that the convergence of the pure-action set of a player in finite-player game $\Gna\esnu$ to that of her corresponding nonatomic player in $\Gna$, assumed by \Cref{eq:def_dset}, can be described by the convergence of spaces of pure-action profiles.  Although $\FX\esnu$ is of finite dimension while $\FX$ is of infinite dimension, the isomorphism $\psi\esnu$ is used to here so that the analysis is done on $L^2([0,1], \rit^T)$.
\begin{lemma}[Convergence of $\FX\esnu$ to $\FX$]\label{lm:FY}
Under \Cref{ass_X_nonat}, for all $\snu\in \nit^*$,\\
(1) for each $\xx\esnu \in \FX\esnu$, $d_2(\psi\esnu(\xx\esnu), \FX)\leq \mdset\esnu$;\\
(2) for each $\xx\in \FX$, $d_\snu(\bpsi\esnu(\xx), \FX\esnu)\leq \mdset\esnu$, $d(\bpsi\esnu(\xx), \FX\esnu)\leq \sqrt{\mmu\esnu}\mdset\esnu$;\\
(3) for each $i\in \I\esnu$ and each $\xx\esnu_i\in \X\esnu_i$, if $d(\frac{\xx\esnu_i}{\mu\esnu_i}, \rbd \frac{\X\esnu_i}{\mu\esnu_i})>\dset\esnu_i$, then $\frac{\xx\esnu_i}{\mu\esnu_i}\in \X_\th$ for all $\th\in \Theta\esnu_i$;\\
(4) for each $i\in \I\esnu$, each $\th\in \Theta\esnu_i$, and each $\xx_\th\in \X_\th$, if $d(\xx_\th, \rbd \X_\th)>\dset\esnu_i$, then $\mu\esnu_i \xx_\th \in \X_i\esnu$.
\end{lemma}
\begin{proof}
(1) Let $\xx\esnu\in \FX\esnu$. For each $i\in \I\esnu$ and each $\th \in \Theta_i\esnu$, define $\yy_\th= \Pi_\th ( \psi\esnu_\th(\xx\esnu )) \in \X_\th$, so that $\|\yy_\th - \psi\esnu_\th(\xx\esnu ) \| \leq \dset_\i\esnu $. Let us show that  $\yy$ is measurable on each $\Theta_i\esnu$, hence measurable on $\Theta$ so that $\yy\in \X$. 
For that,  define $\kappa_i$ on $\Theta\esnu_i\times \M^T$ by $\kappa_i:(\th,\ww)\mapsto \|\xx\esnu_i - \ww\|$. Then, $\kappa_i$ is a Carath\'eodory function. Since the correspondence $\Theta\esnu_i \ni \th \mapsto \X_\th$ is measurable, according to the measurable maximum theorem \cite[Thm. 18.19]{aliprantis2006infinite}, there is a measurable selection of $\xx_\th\in \arg\min_{\X_\th}\kappa_i(\th,\cdot)$. The minimum of $\kappa_i(\th,\cdot)$ on $\X_\th$ is unique and is just $\yy_\th$, hence $\yy$ is measurable on $\Theta_i\esnu$.

Then, $\| \psi\esnu(\xx\esnu ) - \yy\|_2  \leq \mdset\esnu$, which shows that $d_2(\psi\esnu(\xx\esnu ), \FX)\leq \mdset\esnu$. 

(2) 
Let $\xx \in \FX$. For each $\i\in\I\esnu$, $\th\in\Theta\esnu_i$, $\|\xx_\th - \Pi\esnu_{\i}(\xx_\th) \|\leq \dset_i\esnu$.
 Since $\frac{1}{\mu\esnu_i}\X\esnu_{i}$ is a convex subset in $\rit^T$, $\frac{1}{\mu\esnu_i}\int_{\Theta\esnu_i}\Pi\esnu_{i}(\xx_\th)\dth\in \frac{1}{\mu\esnu_i}\X\esnu_{i}$. Define $\yy\in \FX\esnu$ by $\yy_i \eqd \int_{\Theta\esnu_i}\Pi\esnu_{i}(\xx_\th)\dth \in \X\esnu_i$ for $i\in \I\esnu_i$. Then,
\begin{align*}
&\|\bpsi\esnu(\xx)-\yy \|^2_\nu= \sum_{i\in \I\esnu} \frac{1}{\mu\esnu_i}\|\bpsi\esnu_i(\xx)-\int_{\Theta\esnu_i}\Pi\esnu_{i}(\xx_\th)\dth\|^2
=\sum_{i\in \I\esnu} \frac{1}{\mu\esnu_i}\|\int_{\Theta\esnu_i}(\xx_\th - \Pi\esnu_{i}(\xx_\th))\dth\|^2
\\ & \leq \sum_{i\in \I\esnu} \frac{1}{\mu\esnu_i}\mu_i\esnu \int_{\Theta_i}\|\xx_\th - \Pi\esnu_{i}(\xx_\th)\|^2 \dth = \sum_{i\in \I\esnu}  \int_{\Theta_i}\|\xx_\th - \Pi\esnu_{i}(\xx_\th)\|^2 \dth \leq \sum_{i\in \I\esnu}\mu_i\esnu (\dset\esnu_i)^2 \leq (\mdset\esnu)^2,
\end{align*}
so that $\|\bpsi\esnu(\xx)-\yy \|_\nu\leq \mdset\esnu$, and $\|\bpsi\esnu(\xx)-\yy \|\leq \sqrt{\mmu\esnu}\mdset\esnu$.  This concludes the proof.

(3) Fix $\snu\in\nit^*$, $i\in \I\esnu$ and $\th\in \Theta\esnu_i$. Consider $\xx\esnu_i\in \X\esnu_i$ such that $d(\frac{\xx\esnu_i}{\mu\esnu_i}, \rbd \frac{\X\esnu_i}{\mu\esnu_i})>\eta$ for some $\eta>\dset\esnu_i$. 
Assume that $\frac{\xx\esnu_i}{\mu\esnu_i}\notin \X_\th$ i.e. $\big\|\frac{\xx\esnu_i}{\mu\esnu_i}-\Pi_{\X_\th}(\frac{\xx\esnu_i}{\mu\esnu_i})\big\|>0$. Let $\yy\esnu_i = \frac{\xx\esnu_i}{\mu\esnu_i} +\eta \frac{\frac{\xx\esnu_i}{\mu\esnu_i}-\Pi_{\X_\th}(\frac{\xx\esnu_i}{\mu\esnu_i})}{\big\|\frac{\xx\esnu_i}{\mu\esnu_i}-\Pi_{\X_\th}(\frac{\xx\esnu_i}{\mu\esnu_i})\big\|}\in \frac{\X\esnu_i}{\mu\esnu_i}$. Since $\X\esnu_i$ is convex, 
$d(\yy\esnu_i, \X_\th)=\big\|\frac{\xx\esnu_i}{\mu\esnu_i}-\Pi_{\X_\th}(\frac{\xx\esnu_i}{\mu\esnu_i})\big\|+\big\|\yy\esnu_i - \frac{\xx\esnu_i}{\mu\esnu_i}\big\|>\eta> \mdset\esnu$, which contradicts the fact that $d(\X_\th, \frac{\xx\esnu_i}{\mu\esnu_i})\leq \dset\esnu_i$. Hence $\frac{\xx\esnu_i}{\mu\esnu_i} \in \X_\th$.

(4) The proof is similar to that of (3).
\end{proof}

The sets of aggregate-action profiles in $\G\esnu(A\esnu)$ \emph{with aggregative constraints} converges to the set of aggregate-action profiles of the nonatomic game $\G(A)$ \emph{with an aggregative constraint}, as the following lemma says.
\begin{lemma}\label{lem:hausdorff_agg_sets}
Under \Cref{ass_X_nonat},  for $\snu\in \nit^*$, \\
(1) $d_H ( \, \Sxag\esnu,\Sxag \,) \leq \mdset\esnu$;\\
(2) for $\xxag\in \rlt\Sxag$, if $d(\xxag, \rbd \Sxag)>\mdset\esnu$, then $\xxag\in \Sxag\esnu$;
for $\xxag\esnu\in \rlt\Sxag\esnu$, if $d(\xxag, \rbd \Sxag\esnu)>\mdset\esnu$, then $\xxag\in \Sxag$;\\
(3) for $\xxag\in \rlt A$, if $d(\xxag, \rbd A)>D\esnu$, then $\xxag\in \A\esnu$; 
 for $\xxag\esnu\in \rlt A\esnu$, if $d(\xxag\esnu, \rbd A\esnu) > D\esnu$, then $\xxag\esnu \in \A$; \\
(4) for $\xxag\in \rlt(\Sxag\cap A)$, if $d(\xxag, \rbd (\Sxag\cap A)) > \max (\mdset\esnu, D\esnu)$, then $\xxag\in \Sxag\esnu\cap A\esnu$; 
 for $\xxag\esnu\in \rlt(\Sxag\esnu\cap A\esnu)$, if $d(\xxag\esnu, \rbd (\Sxag\esnu \cap A\esnu)) > \max (\mdset\esnu, D\esnu)$, then $\xxag\esnu \in \Sxag\cap A$. 
\end{lemma}
\begin{proof}
\ifproofs
(1)  Define mapping $S\esnu$ from $\M^{I\esnu}$ to $\rit^T$ by $S\esnu (\xx\esnu)=\sum_{i\in \I\esnu}\xx\esnu_i$. Then $S\esnu$ is nonexpansive hence it is continuous. To see this, let $\zz\esnu$ and $\ww\esnu$ be in $\M^{I\esnu}$, then by the Cauchy-Schwarz inequality,
\begin{equation*}
 \big\|\sum_{i\in \I\esnu}\zz\esnu_i - \sum_{i\in\I\esnu}\ww\esnu_i \big\|^2
 =\big\|\sum_{i\in\I\esnu}\frac{\zz\esnu_i - \ww\esnu_i }{ \sqrt{\mu\esnu_i}} \sqrt{\mu\esnu_i}\big\|^2 \leq \sum_{\i\in\I\esnu}\frac{\|\zz\esnu_i - \ww\esnu_i\|^2}{\mu\esnu_i}\cdot \sum_{i\in\I\esnu}\mu\esnu_i= \|\zz\esnu-\ww\esnu \|^2_\nu.
 \end{equation*}

Now fix $\xx \in \FX$. 
Consider $\yy\in \FX\esnu$ such that $\|\bpsi\esnu(\xx)-\yy \|_\nu\leq \mdset\esnu$ (cf. \Cref{lm:FY}). Then $\|\txt\sum_{i\in \I\esnu}\bpsi_i\esnu(\xx) - \txt\sum_{i\in\I\esnu}\yy_i\|^2\leq \|\bpsi\esnu(\xx)-\yy \|^2_\nu \leq (\mdset\esnu)^2$. Hence $d(\int \xx, \Sxag\esnu)\leq \|\txt\sum_{i\in \I\esnu}\bpsi_i\esnu(\xx) - \txt\sum_{i\in\I\esnu}\yy_i\|\leq \mdset\esnu$.

On the other hand, let $\xx\esnu\in \FX\esnu$, thus $\xxag\esnu \eqd \sum_{\i\in\I\esnu} \xx\esnu_i \in \Sxag\esnu$. For each $i\in \I\esnu$ and each $\th \in \Theta_i\esnu$, define $\yy_\th= \Pi_\th\big(\frac{1}{\mu_\i\esnu} \xx\esnu_i \big) \in \X_\th$, so that $\|\frac{1}{\mu_\i\esnu}\xx\esnu_i - \yy_\th \| \leq \dset_\i\esnu $. Then, $\|\sum_{\i\in\I\esnu} \xx\esnu_i  - \int \yy \| \leq \sum_{\i\in\I\esnu}  \|\int_{\Theta_\i\esnu}\frac{1}{\mu_\i\esnu}  \xx\esnu_i  - \yy_\th \dth  \| \leq  \sum_{\i\in\I\esnu}  \int_{\Theta_\i\esnu} \| \frac{1}{\mu_\i\esnu}  \xx\esnu_i  - \yy_\th \| \dth \leq  \sum_{i\in\I\esnu} \mu_\i\esnu \dset_\i\esnu \leq \mdset\esnu $, 
which shows that $d \left( \xxag\esnu, \Sxag\right) \leq \mdset\esnu$ for all $\xxag\esnu\in\ \Sxag\esnu$. 
\smallskip

(2-3)  The proof is similar to that for \Cref{lm:FY}.(3).

(4) These are corollaries of (2) and (3).
%
%
\end{proof}

\begin{lemma}\label{lm:intprofile}
Under \Cref{ass_X_nonat,asp:intpt}, there is a strictly positive constant $\rho$ and a nonatomic pure-action profile $\zz\in \FX$ such that $\int \zz \in \rlt (\Sxag\cap A)$ and, for almost all $\th\in \Theta$, $d(\zz_\th, \rbd \X_\th)>3\rho$.
\end{lemma}
\begin{proof}
Take $\bar{\xx}$ the nonatomic pure-action profile in \Cref{asp:intpt} and an arbitrary $\yy\in\FX(A)$ such that $\int \yy\in \rlt (\Sxag\cap A)$.  

Denote $t=\frac{d(\int\yy, \rbd (\Sxag\cap A))}{3M}$.  Define profile $\zz\in \FX$ by $\zz =\yy - t(\yy -\bar{ \xx}) $. 

Firstly, $\|\int\yy - \int\zz\|=t\|\int\yy-\int\bar{\xx}\|\leq t 2M\leq \frac{2}{3}d(\int\yy, \rbd (\Sxag\cap A))$, hence $\int\zz \in \rlt (\Sxag \cap A)$.

Besides, for any $\th$, $\zz_\th =\yy_\th - t(\yy_\th - \bar{\xx}_\th) $. Since $d(\bar{\xx}_\th, \rbd \X_\th)>\eta$, $\yy_\th\in \X_\th$, and $\X_\th$ is convex, one has $d(\zz_\th, \rbd \X_\th ) > \eta t = \frac{\eta}{3M}d(\int\yy, \rbd (\Sxag\cap A))$. One concludes by defining  $\rho \eqd \frac{\eta}{9M}d(\int\yy, \rbd (\Sxag\cap A))$.
\end{proof}

\begin{notation}
Denote $\zzag=\int \zz$ where $\zz$ is the one in \Cref{lm:intprofile}. Define $\bar{\rho}=\frac{1}{3}d(\zzag, \rbd \Sxag\cap A)>0$. Then,  for $\snu\in \nit^*$ such that $\max(\mdset\esnu, D\esnu)<3\bar{\rho}$,  $d(\zzag, \rbd (\Sxag\esnu\cap A\esnu))  \geq 3 \bar{\rho}-\max(\mdset\esnu, D\esnu)$.
\end{notation}

\begin{figure}[hbtp!]
\centering
\includegraphics[scale=1]{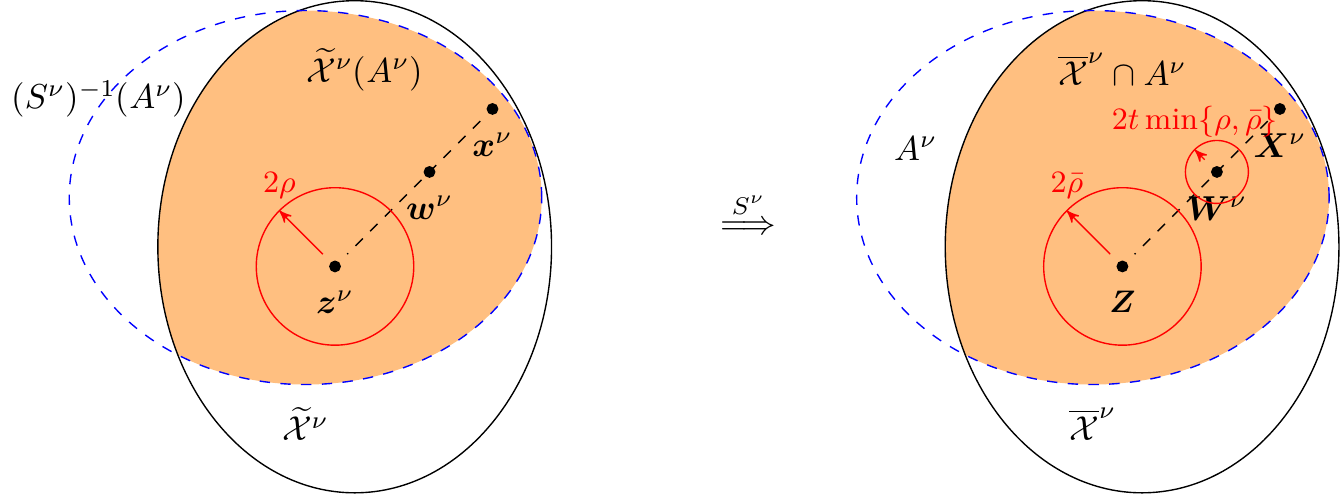}
\caption{Illustration of the mapping $S\esnu$ between $\FX\esnu(A\esnu)$ and $\Sxag\cap A$ used in \Cref{lem:dist_agg_genized_sets}}
\end{figure}

The following lemma shows that the space of pure-action profiles in the finite-player game \emph{with aggregative constraint}, $\FX\esnu(A\esnu)$, is converging to the space of pure-action profiles in the nonatomic game \emph{with aggregative constraint}, $\FX(A)$.
\begin{lemma}[Convergence of $\FX\esnu(A\esnu)$ to $\FX(A)$]
\label{lem:dist_agg_genized_sets}
Under \Cref{ass_X_nonat,asp:intpt}, let $K_A=\tfrac{M+1}{\min\{\rho, \bar{\rho}\}}$. Then, for all $\snu\in \nit^*$ such that $\max(\mdset\esnu, D\esnu) < \min\{\bar{\rho}, \rho\}$,\\
(1) for each $\xx\esnu \in \FX\esnu(A\esnu)$, $d_2(\psi\esnu(\xx\esnu), \FX(A))\leq 2K_{A}\max(D\esnu,\mdset\esnu)$;\\
(2) for each $\xx\in \FX(A)$, $d_{\snu}(\bpsi\esnu(\xx), \FX\esnu(A\esnu))\leq K_{A}\max(D\esnu,\mdset\esnu)$ and 
 $d(\bpsi\esnu(\xx), \FX\esnu(A\esnu)) \leq \sqrt{\mmu\esnu} K_{A}\max(D\esnu,\mdset\esnu)$. 
\end{lemma}
\begin{proof}
(1) Consider $\xx\esnu \in \FX\esnu(A\esnu)$ and $\xxag\esnu=\sum_i  \xx_i\esnu$, i.e. $S\esnu(\xx\esnu)$. 
%
Consider $\zz$ in \Cref{lm:intprofile} and $\zzag=\int \zz$. Let $\zz\esnu\eqd \bpsi\esnu(\zz)$. Since for each $\th$,  $d(\zz_\th,\rbd \X_\th) > 3\rho > \mdset\esnu$, $\zz\esnu \in \FX\esnu$ according to \Cref{lm:FY}.(4).

Also, $d(\frac{\zz_i\esnu}{\mu\esnu_i}, \rbd \frac{\X_i\esnu}{\mu\esnu_i}) \geq 2\rho$. 
Indeed, if $\frac{\yy_i}{\mu\esnu_i} \in N_{2\rho}( \frac{\zz_i\esnu}{\mu\esnu_i})$, then let $\yy_\th \eqd \zz_\th + (\frac{\yy_i}{\mu\esnu_i}-\frac{\zz\esnu_i}{\mu\esnu_i} )$ for $\th\in\Theta_i\esnu$. Then, $\yy_\th \in  N_{2\rho}(\zz_\th)\subset \X_\th$ as $\spa \X_\th = \spa \X_i\esnu$, and $d(\yy_\th, \rbd \X_\th) \geq d(\zz_\th, \rbd \X_\th) -\norm{\zz_\th-\yy_\th} > 3\rho - 2\rho= \rho$. Hence, from \ref{lm:FY}.(4), one has $\frac{\yy_i}{\mu\esnu_i} \in \frac{\X_i\esnu}{\mu\esnu_i}$.

Besides, since $\max(\mdset\esnu, D\esnu) < \bar{\rho}$, one has  $d(\zzag, \rbd (\Sxag\esnu\cap A\esnu))  \geq 2 \bar{\rho}$.

Define $\ww\esnu\eqd\xx\esnu+t(\zz\esnu-\xx\esnu)$ with $t\eqd\frac{\max(D\esnu,\mdset\esnu)}{\min\{\rho, \bar{\rho}\}}<1$. Let $\wwag\esnu=S\esnu(\ww\esnu)$. 
Then, $\|\ww\esnu-\xx\esnu\|_\snu=\max(D\esnu,\mdset\esnu)\frac{\|\zz\esnu-\xx\esnu\|_\nu}{\min\{\rho, \bar{\rho}\}} \leq \max(D\esnu,\mdset\esnu)\frac{2(M+1)}{ \min\{\rho, \bar{\rho}\}} $.

The linear mapping $S\esnu$ maps the segment linking $\xx\esnu$ and $\zz\esnu$ in $\FX\esnu(A\esnu)$ to a segment linking $\xxag\esnu$ and $\zzag$ in $\Sxag\esnu\cap A\esnu$. Hence, by the definition of $\ww\esnu$,
 $N_{2 t \min\{\rho, \bar{\rho}\}}(\wwag\esnu) \cap\spa (\Sxag\esnu\cap A\esnu) \subset \Sxag\esnu\cap A\esnu$, because each point in $N_{2 t \min\{\rho, \bar{\rho}\}}(\wwag\esnu) \cap\spa (\Sxag\esnu\cap A\esnu)$ is on the segment linking $\xxag\esnu$ and some point in $N_{2 \min\{\rho, \bar{\rho}\}}(\zzag) \cap\spa (\Sxag\esnu\cap A\esnu)\subset \Sxag\esnu\cap A\esnu$.
In particular,  $2 t \min\{\rho, \bar{\rho}\}=2 \max(D\esnu,\mdset\esnu)$ means that $d(\wwag\esnu, \rbd(\Sxag\esnu\cap A\esnu))\geq 2\max(D\esnu,\mdset\esnu) >\max(D\esnu,\mdset\esnu)$. Consequently, $\wwag\esnu\in \Sxag\cap A$ according to \Cref{lem:hausdorff_agg_sets}(4). 

Then, in the same manner, one has $d(\frac{\ww\esnu_i}{\mu\esnu_i},\rbd \frac{\X_i\esnu}{\mu_i\esnu})  > 2t\rho> \mdset_\i\esnu $. Thus,  \Cref{lem:hausdorff_agg_sets}(3) implies that $\frac{\ww\esnu}{\mu\esnu_i} \in \X_\th$ for each $\th\in \Theta\esnu_i$, hence $\psi\esnu(\ww\esnu)\in \FX\esnu$ and $\psi\esnu(\ww\esnu)\in \FX\esnu(A)$.

Finally, $\|\xx\esnu-\ww\esnu\|_{\snu} \leq \max(D\esnu,\mdset\esnu)\frac{2(M+1)}{\min\{\rho, \bar{\rho}\}} $. Hence, $d_2(\psi\esnu(\xx\esnu), \FX(A))\leq 2 K_{A}\max(D\esnu,\mdset\esnu)$.

The proof for (2) is similar and omitted.

\end{proof}

\begin{figure}[ht]\label{fig:proj}
\centering
\includegraphics[scale=1]{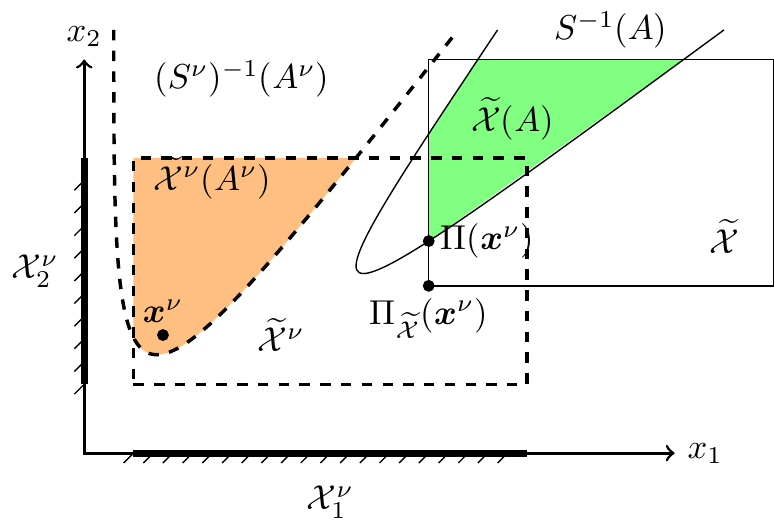}
\caption{Difference between projections on $\FX$ and on $\FX(A)$. (Since it is impossible to draw the graph of a $L^2$ pure-action-profile space with a continuum of players, we illustrate the idea with two players.)}
\end{figure}

\begin{remark}[{Difference between unilateral projections of actions and  collective projection of the action profile}]
\Cref{lem:dist_agg_genized_sets} shows that $d_2(\psi\esnu(\xx\esnu), \Pi(\psi\esnu(\xx\esnu)))\leq 2K_{A}\max(D\esnu,\mdset\esnu)$ and $d(\bpsi\esnu(\xx), \Pi\esnu(\bpsi\esnu(\xx))) \leq \sqrt{\mmu\esnu} K_{A}\max(D\esnu,\mdset\esnu)$. Recall that $\Pi$ and $\Pi\esnu$ stand for the projection functions onto $\FX(A)$ and $\FX\esnu(A\esnu)$. \Cref{lem:dist_agg_genized_sets} is of first importance for our proof of \Cref{thm:converge_with_u}. Without this lemma, we only have the convergence of individual, i.e. unilateral pure-action spaces in the AAS to the unilateral pure-action spaces in the nonatomic game, as shown in \Cref{lm:FY}. Without coupling constraints, this should be sufficient in the proof of the convergence of NE. However, in the presence of coupling aggregative constraints, this convergence of unilateral pure-action spaces is not enough. Given a profile in $\FX\esnu(A\esnu)$, unilateral projection of each atomic player's pure-action onto the corresponding nonatomic players' pure-action spaces, i.e. from $\X\esnu_i$ to $\X_\th$, cannot guarantee that the resulting profile of pure-actions is in $\FX(A)$, and vice versa. \Cref{lem:dist_agg_genized_sets} shows that, for each pure-action profile $\xx\esnu \in \FX\esnu(A\esnu)$, its projection on the space of finite-player pure-action profiles in $\FX(A)$ very close to $\xx\esnu$, and vice versa. 
\end{remark}

We are finally ready to prove \Cref{thm:converge_with_u}.

\begin{proof}[Proof of \Cref{thm:converge_with_u}]
Fix $\snu\in \nit^*$, define  $\hyy\esnu=\psi(\hxx\esnu)\in L^2([0,1], \rit^T)$ and $\hz\esnu\eqd \Pi(\hyy\esnu) $.
  Then $\hz\esnu\in \FX(A)$ is a pure-action profile in nonatomic game $\Gna(A)$. 
By the definition of VWE (cf. \Cref{cond:ind_opt_ve_inf}), there is $\g\in H(\sxx)$ such that $\int_0^1 \langle \g_\th( \sxx_\th, \sxxag) , \ \sxx_\th  - \hz\esnu_\th  \rangle \dth\leq 0$. 

Secondly, by the definition of VNE (cf. \Cref{cond:ind_opt_ve}), there is $\h_i \in \partial_1 \hf\esnu_\i(\hxx\esnu_i, \hxxag\esnu_{-i})$ for each $i\in\I\esnu$ such that $\sum_{\i\in\I\esnu}  \big\langle  \h_\i( \hxx\esnu_\i, \hxxag_{-i}\esnu), \hxx\esnu_\i - \zz\esnu_i\rangle \leq 0$ for all $\zz\esnu_\i \in \X\esnu_\i$. For all $\i\in \I\esnu$ and $\th\in \Theta\esnu_i$, by the definition of $\ld\esnu_\i$ and $\duti\esnu_i$ (cf. \Cref{eq:def_ld,eq:def_dut}), there exists $\rr\in H(\hyy\esnu)$ with $\rr_\th\in \partial_1 f_\th( \hyy\esnu_\th, \hyyag\esnu)=\partial_1 f_\th(\frac{\hxx\esnu_\th}{ \mu\esnu_\i }, \hxxag\esnu) $ such that $ \|\rr_\th( \hyy\esnu_\th, \hyyag\esnu) - \h_\i( \hxx\esnu_\th, \hxxag_{-i}\esnu) \|\leq \ld_i\esnu+\duti\esnu_i$.

Thirdly, $\|\hyy\esnu -\hz\esnu\|_2\leq 2 K_A\max(D\esnu,\mdset\esnu)$ by \Cref{lem:dist_agg_genized_sets}.

With these two results, while noticing that $\hyy\esnu_\th \leq M+\mdset\esnu$ for all $\th$ by \Cref{lem:norm_xat_approx} and that $\hxxag\esnu=\hyyag\esnu$, one has:
\begin{align}
&  \int_{\Theta} \big\langle   \g_\th( \sxx_\th, \sxxag)- \rr_\th( \hyy\esnu_\th, \hyyag\esnu) , \ \sxx_\th-  \hyy\esnu_\th \big\rangle \dth\notag \\
=&  \int_{\Theta} \left\langle \g_\th( \sxx_\th, \sxxag) , \ \sxx_\th  - \hz\esnu_\th  \right \rangle \dth + \int_{\Theta} \left\langle \g_\th( \sxx_\th, \sxxag) , \  \hz\esnu_\th - \hyy\esnu_\th \right \rangle \dth \notag  \\
&\!+ \!\sum_{\i\in\I\esnu}  \int_{\Theta\esnu_i} \big\langle  \rr_\th( \hyy\esnu_\th, \hyyag\esnu) \!- \! \h_\i( \hxx\esnu_\i, \hxxag_{-\i}\esnu)  , \ \hyy\esnu_\th\! -\! \sxx_\th  \big\rangle \dth
\!+\! \sum_{\i\in\I\esnu}  \int_{\Theta\esnu_i} \big\langle  \h_\i( \hxx\esnu_\i, \hxxag_{-i}\esnu)  , \ \hyy\esnu_\th \!-\! \sxx_\th  \big\rangle \dth \notag \\
\leq &\, 0\!+\!\int_{\Theta}  \norm{\g_\th( \sxx_\th, \sxxag)} \norm{\hz_\th\!-\!\hyy\esnu_\th} \!\dth \! + \! \sum_{\i\in\I\esnu} \int_{\Theta\esnu_i} \|\rr_\th( \hyy\esnu_\th, \hyyag\esnu) \!-\! \h_\i( \hxx\esnu_\th, \hxxag_{-i}\esnu) \| \norm{\hyy\esnu_\th\! -\! \sxx_\th }\!\dth\! + \! J\esnu\notag\\
\leq & \, 2 \Bdf  \, K_A\max(D\esnu,\mdset\esnu)  +\,(2M+\mdset\esnu)(\mduti\esnu+\mld\esnu)+J\esnu   \label{eq:th-firstbound}
\end{align}
where $J\esnu\eqd  \sum_{\i\in\I\esnu}  \int_{\Theta\esnu_i} \big\langle  \h_\i( \hxx\esnu_\i, \hxxag_{-i}\esnu) , \ \hyy\esnu_\th - \sxx_\th  \big\rangle \dth$.

Next, for the VWE $\sxx\in\FX(A)$, let
 $\syy\esnu=\bpsi(\sxx)\in  \M^{\I\esnu}$ and  $\sz\esnu \eqd \Pi\esnu(\syy\esnu) \in \FX\esnu(A\esnu)$:
\begin{align}
 J\esnu & =\sum_{\i\in\I\esnu}  \big\langle  \h_\i( \hxx\esnu_\i, \hxxag_{-i}\esnu), \ \hxx\esnu_\i - \syy_i\esnu  \big\rangle \notag \\
 & = \sum_{\i\in\I\esnu}  \big\langle  \h_\i( \hxx\esnu_\i, \hxxag_{-i}\esnu) , \ \hxx\esnu_\i - \sz_i\esnu  \big\rangle + \sum_{\i\in\I\esnu}  \big\langle \h_\i( \hxx\esnu_\i, \hxxag_{-i}\esnu) , \ \sz\esnu_\i - \syy_i\esnu  \big\rangle \notag \\
 & \leq 0 +   (\Bdf+\mld\esnu) \norm{\sz\esnu- \syy\esnu} \leq  (\Bdf+\mld\esnu)\sqrt{\mmu\esnu}K_A\max(D\esnu,\mdset\esnu) \, ,\label{eq:Jnu}
\end{align}
because of the definition of VNE $\hxx\esnu$, the definition of $\mld\esnu$ and  \Cref{lem:dist_agg_genized_sets}.$(i)$.

Let us summarize by combining \eqref{eq:th-firstbound} and \eqref{eq:Jnu}:
\begin{equation}\label{eq:unpperboundmono}
\begin{aligned}  
\int_{\Theta}  \big\langle \g_\th( \sxx_\th, \sxxag) &-\rr_\th( \hyy\esnu_\th, \hyyag\esnu) ,  \sxx_\th-  \hyy\esnu_\th \big\rangle \dth\\
& \leq (3\Bdf + 1) K_A\max(D\esnu,\mdset\esnu)  +(2M+1)(\mduti\esnu + \mld\esnu) \ .
\end{aligned}
\end{equation}
Hence, if $\Gna$ is strongly monotone with modulus $\stgccvut$, then $ \stgccvut  \norm{ \hyy\esnu -  \sxx }^2_2  \leq  (2\Bdf + 1) K_A(D\esnu+\mdset\esnu)  +(2M+1)(\mduti\esnu + \mld\esnu)$. If $\Gna$ is aggregatively strongly monotone with modulus $\beta$, then $ \beta \| \hxxag\esnu - \sxxag\|^2\leq  (3\Bdf + 1) K_A\max(D\esnu,\mdset\esnu)  +(2M+1)(\mduti\esnu + \mld\esnu)$.
\end{proof}

\begin{remark}
The strong monotonicity of the nonatomic game $\Gna$, either with respect to pure-action profile or with respect to aggregate-action profile, is essential in this result. Strict monotonicity is not enough, in contrast to finite-player games (cf. \cite{ wan2012coalition}). Indeed, since $L^2([0,1],\M^T)$ is only weakly compact, one cannot ensure that $\int_{\Theta}  \langle \g_\th( \sxx_\th, \sxxag) -\rr_\th( \hyy\esnu_\th, \hyyag\esnu) ,  \sxx_\th-  \hyy\esnu_\th \rangle \dth$ tends to $\int_{\Theta}  \langle \g_\th( \sxx_\th, \sxxag) -\rr_\th( \hzz_\th, \hat{\zzag}) ,  \sxx_\th-  \hzz_\th \rangle \dth$ in \eqref{eq:unpperboundmono}, where $\hzz$ is an accumulation point of $(\hyy\esnu)_{\snu}$ in the weak topology.
 \end{remark}

\subsection{Construction of an Atomic Approximating Sequence}\label{subsec:construction}

%
As seen in our previous results, a nonatomic player $\th$ is characterized by two elements:  her pure-action set $\X_\th$, and  her subdifferential correspondence $\partial_1 f_\th$ defined from $\M^2$ to $\rit^T$: $(\xx,\yyag)\mapsto \partial_1 f_\th(\xx, \yyag)$.

Note that it is the subdifferential of the cost function $\partial_1 f_\th$, instead of the cost function $f_\th$ itself, that characterizes a nonatomic player's type. For example, two players $\th$ and $\xi$ with $\X_\th=\X_\xi$ and $f_\th(\xx,\yyag) \equiv f_\xi(\xx,\yyag) + C$ where $C$ is a strictly positive constant can be seen as identical in their behavior.

%
%
%
%

This section presents the construction of an AAS for a given nonatomic game $\Gna$ in two particular cases:  (1) the player characteristic profile $\th \mapsto (\X_\th, \partial_1 f_\th)$ is piecewise continuous (cf. \Cref{def:continuity_nonatomic_game}) and, (2) \{$\X_\th,\th\in \Theta\}$ and  \{$f_\th, \th\in \Theta\}$ are respectively polytopes and functions parameterized by a finite number of real parameters.

\paragraph{Case 1: Piecewise Continuous Characteristics -- Uniform Splitting} ~~
\begin{definition}[Continuity of nonatomic player characteristic profile]\label{def:continuity_nonatomic_game}
The player characteristic profile $\th \mapsto (\X_\th, \partial_1 f_\th)$ in nonatomic game $\Gna$ is \emph{continuous} at $\th \in \Theta$ if, for all $\varepsilon>0$, there exists $\eta>0$ such that: for each $\th'\in \Theta$
\begin{equation}\label{eq:character_continu}
 |\th-\th'| \leq \eta \ \Rightarrow
\begin{cases}
 d_H(\X_\th, \X_{\th'}) \leq \varepsilon \, \\
 \sup_{(\xx,\yyag)\in \M\times \M} d_H\left(
 \partial_1 f_\th(\xx,\yyag) ,
  \partial_1 f_{\th'}(\xx,\yyag) \right) \leq \varepsilon  \ .
\end{cases}
\end{equation}

If \eqref{eq:character_continu} is true for all $\th$ and $\th'$ on an interval $\Theta' \subset \Theta$,  then the player characteristic profile is \emph{uniformly continuous} on $\Theta'$. 
\end{definition}

Assume that the player characteristic profile $\th \mapsto (\X_\th, \partial_1 f_\th)$ of nonatomic game $\Gna$ is piecewise continuous, with a finite number $K$ of discontinuity points $\sigma_0=0 \leq \disth_1 <\disth_2 < \dots < \disth_K\leq \sigma_K=1$, and that it is uniformly continuous on $( \sigma_k, \sigma_{k+1})$, for each $k\in  \{0 ,\dots ,K-1\}$.

For $\nu \in \nit^*$, define an ordered set of $I_\snu$ cutting points by $
\{\cutp_i\esnu, i=0,\ldots,I\esnu\} :=\left\{ \tfrac{k}{\nu} \right\}_{0\leq k \leq \nu} \cup \{\sigma_k   \}_{1\leq k \leq K}$
 and the corresponding partition $(\Theta_\i\esnu)_{i\in \I\esnu}$ of $\Theta$ by:
\vspace{-0.15cm}
\begin{equation*}
\Theta_i\esnu= [\cutp_{i-1}\esnu,\cutp_{i}\esnu ) \text{ for }  \i \in \{1, \dots, I\esnu-1\}\  ;\quad \Theta_{ I\esnu}\esnu = [\cutp_{ I_\snu-1}\esnu, 1 ].
\end{equation*}
Hence, $\mu_i\esnu=\cutp_{i}\esnu-\cutp_{i-1}\esnu$. Denote $\bar{\cutp}_{i}\esnu= \frac{\cutp_{i-1}\esnu +  \cutp_{i}\esnu}{2}$.

 \begin{assumption}\label{asp:partials} 
%
For each $\th\in \Theta$ and each $\xx\in \M$, $\M'\ni \yyag \mapsto f_\th(\xx,\yyag)$ is convex. Denote the subdifferential of $f_\th(\xx,\cdot)$ at $\yyag$ by $\partial_2 f_\th(\xx,\yyag)$. There is $\Bdg>0$ such that $\tnorm{\g} \leq \Bdg $ for all $\g\in \partial_2 f_\th(\cdot, \cdot)$ on $\M^2$ for all $\th$.
\end{assumption}
\begin{proposition}\label{prop:approx_seq_continuous}
Under \Cref{ass_X_nonat,ass_ut_nonat,assp:tech,assp_lip,asp:partials}, for $\snu \in \nit^*$, consider the finite-player game $\G\esnu(A\esnu)$ with an aggregative constraint defined with $A\esnu\eqd A$, player set $\I\esnu\eqd\{ 1  \dots  I\esnu\}$, where for each finite player $\i\in\I\esnu$:
\vspace{-0.15cm}
\begin{equation*}
\X_i\esnu \eqd \mu_i\esnu \X_{ \bar{\cutp}_{i}\esnu} 
\text{ \ and \ } f_i\esnu(\xx,\yyag)\eqd \mu_i\esnu f_{ \bar{\cutp}_{i}\esnu }\Big(\txt\frac{1}{\mu_i\esnu} \xx ,\yyag\Big) ,\;  \, \forall (\xx,\yyag)\in \mu_i\esnu\M \times \M.
\end{equation*}
Then $\big(\G\esnu(A)\big)_{\snu}=\big( \I\esnu, \FX\esnu,A, (f\esnu_\i)_{\i\in \I\esnu} \big)_{\snu} $ is an AAS of nonatomic game  $\Gna(A)$. 
\end{proposition}
\begin{proof} 
First notice that $\mmu\esnu \leq \frac{1}{\nu} $ hence it tends to 0. 

Let us show the four points required by \Cref{def:approx_seq} as follows.
%
%

i) Given an arbitrary $\varepsilon>0$, there is a common modulus of uniform continuity $\eta$ such that \Cref{eq:character_continu} is true for all the intervals $(\sigma_k,\sigma_{k+1})$. For $\nu$ large enough, one has, for each $\i \in\I\esnu$, 
 $\mu_\i\esnu< \eta$ so that for all $\th \in \Theta_\i\esnu$, $|\bar{\cutp}_{i}\esnu-\th|< \eta$; hence $d_H\big( \X_\th, \txt\frac{1}{\mu_\i\esnu}\X_\i\esnu\big)=d_H( \X_\th, \X_{ \bar{\cutp}_{i}\esnu}) < \varepsilon$. 

ii) According to the continuity property, for all $(\xx,\yyag)\in \M^2$:
\begin{equation*}
d_H\left( \partial_1 f\esnu_\i\left(\txt {\mu_i\esnu } \xx, \yyag \right) \ ,
\  \partial_1 f_\th(\xx,\yyag) \right)  = d_H\left( \txt\frac{\mu_\i\esnu}{\mu_i\esnu}   \partial_1 f_{ \bar{\cutp}_{i}\esnu}\Big(\txt\frac{1}{\mu_i\esnu} \mu_\i\esnu \xx , \yyag \Big) \ , \ \partial_1 f_\th(\xx,\yyag)  \right)< \varepsilon . 
\end{equation*}
To be rigorous, we need $\spa \X_\th$ to be the same for all $\th\in \Theta\esnu_i$. For this, we can further divide $\Theta\esnu_i$ into a finite number of, say $n$, groups so that the nonatomic players in each group have the same $\spa \X_\th$. This is possible because $\X_\th$ are all in $\rit^T$, a finite dimensional space. 

iii)  Consider $\partial_1 \hf\esnu_\i(\xx, \yyag-\xx)$, the subdifferential of $\hf\esnu_\i(\cdot, \yyag')$ at $\xx$, given $\yyag'=\yyag-\xx$. By definition, $\hf\esnu_\i(\cdot, \yyag')\eqd  f\esnu_\i(\cdot, \yyag'+\cdot)$. Consider $\hf\esnu_\i(\cdot, \yyag')$ as a composition function on $\M'$ defined by $f\esnu_\i \circ L$, where $L$ is the following affine mapping from $\mu\esnu_\i\M'$ to $\rit^{2T}$:
\begin{equation*}
L(\xx) = \begin{pmatrix} I_{T} \\ I_{T}\end{pmatrix} \xx + 
\begin{pmatrix} \mathbf{0}_{T} \\ \yyag' \end{pmatrix}\ ,
\end{equation*}
where $I_{T}$ is the $T$-dimensional identity matrix and $\mathbf{0}_T$ the $T$-dimensional 0 vector. 

On the one hand, according to \cite[Proposition 16.6]{combettes2011monotone}, $\partial_1 \hf\esnu_\i(\xx, \yyag-\xx)\subset \{(I_T,I_T) \g :\g\in \partial f\esnu_\i(\xx, \yyag)\}$, where $\partial f\esnu_\i(\xx, \yyag)$ is the subdifferential of $f\esnu_\i$ at $(\xx, \yyag)$. On the other hand, according to \cite[Proposition 16.7]{combettes2011monotone}, 
\begin{equation*}
\partial f\esnu_\i(\xx, \yyag)\subset 
\left\{
\begin{pmatrix} \g_1 \\ \g_2 \end{pmatrix} : 
\g_1 \in\partial_1 f\esnu_\i(\xx, \yyag), \,
\g_2 \in \partial_2 f\esnu_\i(\xx, \yyag)\right\}\ .
\end{equation*}
Therefore, $\partial_1 \hf\esnu_\i(\xx, \yyag-\xx)\subset \{\g_1 + \g_2: \g_1 \in\partial_1 f\esnu_\i(\xx, \yyag), \,
\g_2 \in \partial_2 f\esnu_\i(\xx, \yyag)\}$. 

By the definition of $f\esnu_i$, 
at $(\xx,\yyag)\in \mu\esnu_\i \M \times \M$, the subdifferential of $f\esnu_\i$ with respect to the second variable is $\partial_2 f_\i\esnu(\xx,\yyag)=\mu\esnu_\i f_{ \bar{\cutp}_{i}\esnu }(\txt\frac{1}{\mu_i\esnu} \xx ,\yyag)$, hence each $\g_2 \in \partial_2 f\esnu_\i(\xx, \yyag)$ is bounded by $\mu\esnu_\i \Bdg$ in $l^2$-norm according to \Cref{asp:partials}. In consequence, for each $\h\in \partial_1 \hf\esnu_\i(\xx, \yyag-\xx)$, there is $\g_1 \in\partial_1 f\esnu_\i(\xx, \yyag)$ such that $\|\h-\g_1\|\leq \mu\esnu_\i \Bdg$. 
Finally, set $\ld\esnu_i \eqd \mu\esnu_\i \Bdg$,  and $\mld\esnu \eqd \max_i \ld\esnu_\i$ tends to 0.

iv) By definition, $D\esnu=0$.
\end{proof}

\paragraph{Case 2: Finite-dimensions Parameterized Characteristics -- Meshgrid Approximation}~~ 

Assume that nonatomic game $\Gna$ satisfy two conditions:

(i) The feasible pure-action sets are $\dimp$-dimensional polytopes: there exist a constant real-valued $\dimp \times T$ matrix $\A$, and a bounded mapping $\bb:\Theta\rightarrow\rit^\dimp$, such that for any $\th$, $\X_\th = \{\xx \in \rit^T  :  \A\xx \leq \bb_\th \}$, which is a nonempty, bounded, closed and convex polytope in $\rit^T$.

(ii) There is a bounded mapping $\ss:\Theta \rightarrow \rit^l$ 
such that for any $\th \in \Theta, f_\th(\cdot,\cdot)= f(\cdot,\cdot\,;\ss_\th)$.   Furthermore, for all $(\xx,\yyag)\in \M^2$, $\partial_1 f(\xx,\yyag;\cdot)$ is Lipschitz-continuous in $\ss$ and with a Lipschitz constant $L_3$, independent of $\xx$ and $\yyag$.

Denote $\ub_k=  \min_\th b_{\th,k} $, $\ob_k = \max_\th b_{\th,k}$ for $k \in \{1\dots \dimp \}$ and $\us_k= \min_\th s_{\th,k} $, $\os_k = \max_\th s_{\th,k}$ for $k\in \{1\dots l\}$.  The characteristics of player $\th$ are parameterized by point $(\bb_\th, \ss_\th)$ in $\prod_{k=1}^\dimp [\ub_k,\ob_k] \times \prod_{k=1}^l [\us_k,\os_k]$, a compact subset of $\rit^{\dimp+l}$.

Fix $\nu \in \nit^*$, consider a uniform partition of the compact set $\prod_{k=1}^\dimp [\ub_k,\ob_k] \times \prod_{k=1}^l [\us_k,\os_k]$, obtained by dividing each dimension of this compact set into $\snu$ equal parts. Hence, the partition is composed of $I\esnu \eqd \nu^{\dimp+l}$ equal-sized subsets of $\prod_{k=1}^\dimp [\ub_k,\ob_k] \times \prod_{k=1}^l [\us_k,\os_k]$. The  cutting points of the partition are  $\ub_{k,n_k}\eqd \ub_k+\frac{n_k}{\nu}(\ob_k-\ub_k) $ for $k\in \{1,\ldots, \dimp\}$, and $\us_{k,n_k}\eqd \us_k+\frac{n_k}{\nu}(\os_k-\us_k)$ for $k\in \{1,\ldots, l\}$, with $n_k\in \{0,\dots, \nu \}$. Let the set of \emph{vectorial} indices, indexing the partition, be denoted by:
\begin{equation*}
\Gamma\esnu \eqd \{\nn=(n_k)_{k=1}^{\dimp+l} \in \nit^{\dimp+l} \,|\, n_k\in \{1,\ldots, \nu\}\}\ .
\end{equation*}
Define the corresponding partition of the interval $\Theta$: $\Theta=\dot{\bigcup}_{\nn\in \Gamma\esnu} \Theta\esnu_{\nn}$, where:
\begin{align*}
\Theta\esnu_{\nn} \eqd  & \Big\{\th\in \Theta : b_{\th,k} \in [\ub_{k,n_k- 1},\ub_{k,n_k} [ \text{ for } 1\leq k \leq \dimp; \,   s_{\th, k}  \in [\us_{k,n_k-1},\us_{k,n_k}  [ \text{ for } 1\leq k \leq l \Big\}. 
\end{align*}
To be rigorous, when $\ub_{k,n_k}=\ob_{k}$ or $\us_{k,n_k} =\os_k$, the parameter interval is closed at the right. 

Finally, define the set of players $\I\esnu$ as the elements $\nn$ in $\Gamma\esnu$ such that $\mu(\Theta\esnu_{\nn}) >0$.

\begin{remark} \label{rm:approx_finitedim_infininitesimal}
If there is a set of players of strictly positive measure sharing the same $\bb$ and $\ss$, then condition $\max_{\i\in\I\esnu} \mu_\i \rightarrow 0$ is satisfied. In this case, adding another dimension in the partition by cutting $\Theta=[0,1]$ into $\nu$ uniform segments solves the problem.
\end{remark}

\begin{proposition}\label{prop:approx_seq_finitedim}
For $\snu \in \nit^*$, let finite-player game $\G\esnu(A\esnu)$ with an aggregative constraint be defined by $A\esnu\eqd A$, player set $\I\esnu\eqd\{ \nn\in \Gamma\esnu : \mu(\Theta\esnu_{\nn})>0\}$ and, for each finite player $\nn\in\I\esnu$, 
\begin{align*}
 \X_{\nn}\esnu &\eqd  \{\xx \in \rit^T  | \A\xx \leq \txt\int_{\Theta\esnu_{\nn}}\bb_{\th}\, \text{d}\th \}\ , \\
  f_{\nn}\esnu(\xx,\yyag) &\eqd  \mu_{\nn}\esnu f\big( \txt\frac{1}{\mu_{\nn}\esnu} \xx ,\yyag ;\txt\frac{1}{\mu_{\nn}\esnu}\txt\int_{\Theta_{\nn}\esnu} \ss_{\th} \dth \big)  , \quad \forall (\xx,\yyag)\in \mu_i\esnu\M\times \M.
\end{align*}
Then, under \Cref{ass_X_nonat,ass_ut_nonat,assp:tech,assp_lip,asp:partials},  $(\G\esnu(A))_{\snu}=\big( \I\esnu, \FX\esnu,A, (f\esnu_\i)_{\i\in \I\esnu} \big)_{\snu} $
 is an AAS of the  nonatomic game $\Gna(A)$ with an aggregative constraint. 
\end{proposition}
\begin{proof} 
Let us show the four properties required by \Cref{def:approx_seq} as follows.

%

i) For each $\nn\in\I\esnu$, $\txt\frac{1}{\mu_{\nn}\esnu}\X_{\nn}\esnu =\left\{\xx \in \rit^T  : \A\xx \leq \txt\frac{1}{\mu_{\nn}\esnu}\txt\int_{\Theta\esnu_{\nn}}\bb_{\th}\, \dth \right\}$. Then, by a result generalized from \cite[Thm. 4.1]{batson1987combinatorial}, 
 there is a constant $C_0$ such that, for each $\th'\in \Theta_{\nn}\esnu$: $d_H\left( \X_{\th'}, \txt\frac{1}{\mu_{\nn}\esnu}\X_{\nn}\esnu \right) \leq C_0 \norm{ \bb_{\th'}-  \frac{1}{\mu_{\nn}\esnu}\int_{\Theta\esnu_{\nn}}\bb_{\th}\, \text{d}\th }
\leq \frac{C_0}{\nu} \norm{ \bm{\ob} -\bm{\ub} }$. Hence, $\mdset\esnu$ tends to 0.

ii) For each  $\nn \in \I\esnu$ and each $\th'\in \Theta_{\nn}\esnu$,  for all $(\xx,\yyag)\in \M^2$, one has:
\begin{align*}
 d_H\Big( \partial_1 f\esnu_{\nn} & (\txt {\mu_{\nn}\esnu } \xx, \yyag) \ , \ \partial_1 f_{\th'}(\xx,\yyag) \Big) =    d_H \Big(  \partial_1 f\big(  \xx, \yyag ; \txt\frac{1}{\mu_{\nn}\esnu}\txt\int_{\Theta_{\nn}\esnu} \ss_{\th} \dth\big) \ ,\ \partial_1 f(\xx,\yyag;\ss_{\th'}) \Big) \\
 &\leq L_3 \|\frac{1}{\mu_{\nn}\esnu}\int_{\Theta_{\nn}\esnu} \ss_{\th} \dth - \ss_{\th'} \| \leq  \frac{L_3}{\snu} \norm{\bm{\os}-\bm{\us} } \ ,
\end{align*} 
by the Lipschitz continuity of $\partial_1 f(\xx,\yyag;\cdot)$. Hence, $\mduti\esnu$ tends to 0. 
 
iii) $\mld\esnu$ tends to 0 is proved as in  \Cref{prop:approx_seq_continuous};

iv) By definition,  $D\esnu=0$.
\end{proof}

\begin{remark}  
In \Cref{prop:approx_seq_finitedim}, instead of the average value of the characteristics of nonatomic players on $\Theta_{\nn}\esnu $, one can use the characteristic value of any nonatomic player in $\Theta_{\nn}\esnu$. 
\end{remark}

\begin{remark}
If the finite-player games constructed by these two constructions do not satisfy the condition that $\hf\esnu_i(\cdot,\yyag_{-i})$ is convex in $\xx_i$, then one can use the pseudo-VNE of $\G\esnu(A\esnu)$ (cf.  Appendix). The existence of pseudo-VNE is guaranteed by our construction of $f\esnu_i$'s. 

Recall that most of the aggregatively games have cost functions of form \Cref{eq:common_form_cost}. As long as the conditions in \Cref{lm:monotonemap} (monotone per-unit cost of public products and concave private utility) are satified, $\hf\esnu_i(\cdot,\yyag_{-i})$ is convex in $\xx_i$.
\end{remark}

\begin{remark}
In both constructions above, the number of finite players $I\esnu$ tends to infinite and their ``size'' $\mu\esnu_i$ tends to zero. These conditions are not required in the definition of AAS. The two constructions are made in this way for two reasons: firstly, nonatomic players' characteristics (the pure-action set and the subdifferential of their cost function with respect to the first variable) in a very small set $\Theta\esnu_i$ are ensured to be close enough by construction; secondly, the subdifferential of cost function $f_i$ of finite player $i$ with respect to the second variable $\xxag$ is vanishing so that she behaves almost like a nonatomic player. If we consider pseudo-VNE instead of NE in the approximating finite-player games, then the second reason  no longer exists. If the nonatomic players are fairly homogeneous or even there are finitely-many different types of them, then there is no need to divide $\Theta$ into smaller and smaller intervals to regroup sufficiently homogeneous nonatomic players, as imposed by the first reason. We shall bear in mind these two points when constructing an AAS for explicit games.   
\end{remark}

\section{Conclusion}\label{sec:conclusion}

\Cref{thm:converge_with_u} provides a theoretical basis for the use of finite-dimensional VNE/NE as an approximation of the VWE/WE in a strongly monotone or aggregatively strongly monotone nonatomic aggregative game with or without aggregative constraints. There are numerous research themes related to this result and our topic in general.

Firstly, one needs to find efficient algorithms for the computation of NE and VNE in finite-player aggregative games with or without aggregative constraints via the solution of their characteristic variational inequalities. An extensive literature exists in this regard but our particular case of aggregative game with aggregative constraints may lead to special methods or improvements on existing results \cite{gramm2017}.

Secondly, the extension of evolutionary dynamics for population games and related algorithms to nonatomic games with infinitely many classes of players can be non trivial. A recent work \cite{Hadikhanloo2017} proposes online learning methods for population games with heterogeneous convex pure-action sets. The presence of aggregate constraints adds two additional difficulties as analyzed in \Cref{rm:unilateral_deviation,rm:collective_deviation}. Evolutionary dynamics in population games are based on unilateral adaptations from players. On the one hand, in the presence of coupling constraints,  unilateral deviations by players may well lead to a pure-action profile violating the coupling constraint. On the other hand, a feasible deviation in the pure-action profile cannot always be decomposed into unilateral deviations of players. 

Thirdly, our results are limited to monotone games and the convergence result is limited to strongly monotone games. The study of nonatomic games that are not monotone needs probably other approaches. Indeed, even for population games where there are only finitely many types of players, there exist much fewer results for games that are not linear, potential or monotone.

Fourthly, other methods to deal with aggregate constraints exist, such as rationing in the case of capacity constraint of network or power grid.


\section*{Acknowledgments}
\noindent We thank Sylvain Sorin, St\'ephane Gaubert, Marco Mazzola, Olivier Beaude and Nadia Oudjane for their insightful comments.

\section*{Appendix: Nonatomic behavior of finite players}\label{app:pseudoVNE}
In the definition of AAS, we suppose that the cost functions of a finite player in a finite-player game is convex in her own action, i.e. $\hf\esnu_{i}(\cdot,\yyag_{-i})$ is convex on $\M$ for each fixed aggregate profile of her rivals $\xx_{-i}$. However, this condition can be replaced by the condition that $f\esnu_{i}(\cdot,\yyag)$ is convex on $\M$ for each fixed aggregate profile $\yyag\in \M$, which is naturally satisfied with the two constructions presented in \Cref{subsec:construction}. In this case, instead of considering VNE in the AAS, we consider another equilibrium notion called pseudo-VNE, where finite players take themselves as nonatomic ones by ignoring the impact of their action on the aggregate-action profile. 
\medskip
 
Consider a finite-player aggregative game $\GA$ with an aggregative constraint defined in \Cref{def:general-atomicGame}.
Define a correspondence $H':\FX\rightrightarrows \rit^{IT}$ as follows: for all $\xx\in \FX$, $H'(\xx) \eqd\{(\g_i)_{i\in \I}\in \rit^{IT}: \g_i \in \partial_1 f_\i(\xx_\i, \xxag) , \  \forall i\in \I\}$,
where $\xxag=\sum_{i\in \I}\xx_i$ is the aggregate-action profile induced by $\xx$.

\begin{definition}\label{def:pseudoVNE}
A \emph{pseudo-VNE} of a finite-player game $\GA$ with an aggregative constraint  is  a solution to the following GVI:
\begin{equation*}
 \text{Find } \hxx\in \FX(A) \text{ s.t. } \exists \g\in H'(\xx) \text{ s.t. } \langle \g, \xx-\hxx \rangle\geq 0,\; \forall \xx\in \FX(A)\ .
 \end{equation*}
\end{definition}

\begin{remark}[Pseudo-VNE is not an equilibirum]
A pseudo-NE/pseudo-VNE is \emph{not} an equilibrium from a game theoretical point of view. Defined as a solution to a certain GVI, it is an auxiliary notion for the approximation of NE/VNE.

The authors in \cite{gentile2017nash} refer to this equilibrium notion as  \textit{Wardrop equilibrium}. We avoid this name here in order to distinguish it from the WE in nonatomic games.
\end{remark}

\begin{assumption}[Convex pseudo-costs] \label{assp_convex_pseudocosts}
For each $\i\in\I$, the function $f_i(\xx_i,\yyag)$ is continuous in both variables, and is convex in $\xx_i$ for all $\yyag\in \Sxag$. 
\end{assumption}
Under \Cref{assp_convex_pseudocosts}, the existence of a pseudo-VNE can be proved as for VNE in \Cref{prop:exist_ve}.
\medskip

Define a \emph{pseudo-AAS} for a nonatomic aggregative game $\GnaA$ with an aggregative constraint be defined almost exactly like an AAS, except that condition (iv) is no longer required, and \Cref{assp_convex_costs} is replaced by \Cref{assp_convex_pseudocosts}. Then, in \Cref{thm:converge_with_u} one can replace the VNE by the pseudo-VNE.
\begin{theorem}[Convergence of pseudo-VNE to VWE] \label{thm:converge_pseudo}
Under \Cref{assp_convex_pseudocosts,assp_compactness,ass_X_nonat,ass_ut_nonat,assp:tech,assp_lip}, let $(\G\esnu(A\esnu))_\snu$ be a pseudo-AAS of  $\Gna(A)$. 
 Let $\sxx$ be the VWE of $\GnaA$, $\hxx\esnu \in \FX\esnu(A\esnu)$ a pseudo-VNE of $\G\esnu(A\esnu)$  for each $\snu\in\nit$, and $\sxxag$, $\hxxag\esnu$ their respective aggregate-action profiles. 
 Then, there exist constants $\rho>0$, $\bar{\rho}$ and $K_A\eqd \tfrac{M+1}{\min\{\rho, \bar{\rho}\}}$ for the following results to hold:

(1) If $\Gna$ is aggregatively strongly monotone with modulus $\beta$, then $(\hxxag\esnu)_\snu$ converges to  $\sxxag$:  for all $\snu\in \nit^*$ such that $2\max(\mdset\esnu, D\esnu) <\rho$, $\| \hxxag\esnu- \sxxag \|^2 \leq \frac{1}{\beta}  \big(  (3\Bdf + 1) K_A\max(D\esnu,\mdset\esnu) +(2M+1)\mduti\esnu \big)$.

(2) If $\Gna$ is strongly monotone with modulus $\alpha$, then 
 $(\psi\esnu(\hxx\esnu))_\snu$  converges to $\sxx$ in $L^2$-norm: for all $\snu \in \nit^*$ such that $2\max(\mdset\esnu, D\esnu) <\rho$, $\|\psi\esnu(\hxx\esnu)-\sxx\|^2_2\leq \frac{1}{ \stgccvut} \big( (3\Bdf +1) K_A\max(D\esnu,\mdset\esnu)  +(2M+1)\mduti\esnu \big)$.

If there are no aggregate constraint, replace $K_A$ and $D\esnu$ all by 0.
\end{theorem}

The proof is almost the same as for \Cref{thm:converge_with_u}. The only difference is that, instead of considering $\h_i \in \partial_1 \hf\esnu_\i(\hxx\esnu_i, \hxxag\esnu_{-i})$ for each $i\in\I\esnu$ such that $\sum_{\i\in\I\esnu}  \big\langle  \h_\i( \hxx\esnu_\i, \hxxag_{-i}\esnu), \hxx\esnu_\i - \zz\esnu_i\rangle \leq 0$ for all $\zz\esnu_\i \in \X_\i$, one has to consider $\h_i \in \partial_1 f\esnu_\i(\hxx\esnu_i, \hxxag\esnu)$ for each $i\in\I\esnu$ such that $\sum_{\i\in\I\esnu}  \big\langle  \h_\i( \hxx\esnu_\i, \hxxag\esnu), \hxx\esnu_\i - \zz\esnu_i\rangle \leq 0$ for all $\zz_\i \in \X_\i$.

\bibliographystyle{ims}


\bibliography{../../bib/longJournalNames,../../bib/biblio1,../../bib/biblio2,../../bib/biblio3,../../bib/biblioBooks}

\end{document}